\documentclass[11pt,dvipsnames,twocolumn]{extarticle}

\usepackage{amsfonts}
\usepackage[leqno]{amsmath}
\usepackage{amssymb}
\usepackage{amsthm}
\usepackage{graphicx}
\usepackage{wasysym}
\usepackage{xcolor}
\usepackage{tikz,tikz-cd}
\usepackage[lmargin=1.25in,rmargin=1.25in,tmargin=1in,bmargin=1in]{geometry}
\usepackage{hyperref}
\usepackage[utf8]{inputenc}
\usepackage[T1]{fontenc}
\usepackage[nameinlink, noabbrev, capitalize]{cleveref}
\usepackage{mathtools} % For \coloneqq
\usepackage[shortlabels]{enumitem}

\usepackage[backend=biber,style=alphabetic,backref=true,sorting=anyt,maxalphanames=3,minalphanames=2,maxbibnames=99,doi=false,isbn=false,url=false,eprint=false]{biblatex}

\renewbibmacro{in:}{}
\DeclareFieldFormat
[article,inbook,incollection,inproceedings,patent,thesis,unpublished]
{title}{#1}
\DefineBibliographyStrings{english}{
	backrefpage={p.},
	backrefpages={pp.}
}
\hypersetup{
	colorlinks=true, urlcolor=NavyBlue, linkcolor=Mahogany, citecolor=ForestGreen, pdfborder={0,0,0},
}

\newcommand{\mb}{\mathbb}

\newcommand{\mc}{\mathcal}

\newcommand{\tbf}{\textbf}
\newcommand{\tsf}{\textsf}
\newcommand{\ttt}{\texttt}
\newcommand{\mbf}{\mathbf}

\newcommand{\mtt}{\mathtt}
\newcommand{\n}{\enspace}
\newcommand{\tx}{\text}
\newcommand{\ol}{\overline}

\newcommand{\supp}{\tx{supp}}

\newcommand{\poly}{\tx{poly}}
\newcommand{\out}{\tx{\normalfont out}}
\newcommand{\inn}{\tx{\normalfont in}}

\newcommand{\F}{\mathbb{F}}

\newcommand{\iref}[2]{(\hyperref[#2]{\ref*{#1}.\ref*{#2}})}

\newcommand{\AEL}{\tx{\normalfont AEL}}

\renewcommand{\epsilon}{\varepsilon}
\newcommand{\eps}{\epsilon}
\newcommand{\Rcor}{\ensuremath{R^*_{\mathrm{cor}}}}
\newcommand{\Rera}{\ensuremath{R^*_{\mathrm{era}}}}

\DeclareUnicodeCharacter{1F600}{\smiley}

\mathchardef\mhyphen="2D

\theoremstyle{theorem}
\newtheorem{theorem}{Theorem}[section]

\newtheorem*{claim*}{Claim}

\newtheorem{proposition}[theorem]{Proposition}

\newtheorem{question}[theorem]{Question}

\theoremstyle{definition}
\newtheorem{remark}[theorem]{Remark}

\newtheoremstyle{TheoremNum}
{\topsep}{\topsep}              %%% space between body and thm
{\itshape}                      %%% Thm body font
{}                              %%% Indent amount (empty = no indent)
{\bfseries}                     %%% Thm head font
{.}                             %%% Punctuation after thm head
{ }                             %%% Space after thm head
{\thmname{#1}\thmnote{ \bfseries #3}}%%% Thm head spec
\theoremstyle{TheoremNum}

\makeatletter
\newcommand{\leqnomode}{\tagsleft@true\let\veqno\@@leqno}
\newcommand{\reqnomode}{\tagsleft@false\let\veqno\@@leqno}
\makeatother

\title{List Recoverable Codes:\\The Good, the Bad, and the Unknown (hopefully not Ugly)}
\author{Nicolas Resch\thanks{University of Amsterdam. Email:~\ttt{n.a.resch@uva.nl}}
	\and S. Venkitesh\thanks{Tel Aviv University.  Email:~\ttt{venkitesh.mail@gmail.com}}}
\date{}
\begin{document}
	
	\maketitle
	
	{\small
	\begin{abstract}
		List recovery is a fundamental task for error-correcting codes, vastly generalizing unique decoding from worst-case errors and list decoding. Briefly, one is given ``soft information'' in the form of input lists $S_1,\dots,S_n$ of bounded size, and one argues that there are not too many codewords that agree a lot with this soft information.  This general problem appears in many guises, both within coding theory and in theoretical computer science more broadly.
		
		In this article we survey recent results on list recovery codes, introducing both the \emph{good} (i.e., possibility results, showing that codes with certain list recoverability exist), the \emph{bad} (impossibility results), and the \emph{unknown}. We additionally demonstrate that, while list recoverable codes were initially introduced as a component in list decoding concatenated codes, they have since found myriad applications to and connections with other topics in theoretical computer science.  
	\end{abstract}
	
	\tableofcontents}
	
	% \begin{itemize}
		% 	\item New sections: ``Large Alphabet Regime'', then ``Small Alphabet Regime'', then ``List-Recovery Beyond Coding Theory''. Subsections for latter: Pseudorandomness \& Leakage-Resilience of Secret-Sharing Schemes.
		% 	\item Nic: I will work on Small Alphabet Regime and LR of SS schemes.
		% 	\item Venki: The complement
		% 	\item Story: up until recently, almost all combinatorial upper bounds were also algorithmic -- they came with decoding algorithms. While these results are very interesting, they are often far from what we can reach combinatorially. Only in recent years have direct combinatorial techniques been devised allowing for a deeper probing of what can and cannot be done. 
		% \end{itemize}
	
	\newpage
	\section{Introduction}\label{sec:intro}

	At its core, coding theory is concerned with the following basic question: what is the best possible tradeoff between a code's rate (which quantifies the efficiency of communication) and its noise-resilience. Here we mean ``noise-resilience'' in a very broad sense: essentially, any forms of errors that many be introduced to data, whether these are performed maliciously or stochastically. For example, if one is concerned with adversarial symbol corruptions or erasures, then minimum distance tightly characterizes the code's fault-tolerance capabilities, at least if one insists that one can always perfectly recover the data (namely, if one insists on \emph{unique decoding}). While this model is already very interesting, in this survey we will concern ourself with a broad generalization, called \emph{list recovery}. 
	% \closed{\nic{edited this, hopefully it strikes the correct balance between abstraction and clarity}\venki{Looks good.}}
	% \closed{\venki{How is `error-correction' different from `decoding from errors'?  Of course, the term `error-correction' is used usually in the context of the rate v/s distance tradeoff, but I've always wondered what do we really mean by correcting errors unless we deduce the capacity theorem.}}
	
	\subsection{What?}
	
	Before continuing the discussion, let us fix some notation. Let \(\Sigma\) be a finite alphabet.  We assume the \tsf{(relative) Hamming distance} on \(\Sigma^n\), that is, for any two strings \(u,v\in\Sigma^n\), the distance between \(u\) and \(v\), denoted by \(\Delta(u,v)\), is the fraction of coordinates \(i\in[n]\) on which \(u\) and \(v\) differ.  A \tsf{code} is simply a subset \(C\subseteq\Sigma^n\). The quantity \(R(C)\coloneqq\frac{\log_{|\Sigma|}|C|}{n}\) is the \tsf{rate} of \(C\), and the \tsf{minimum (relative Hamming) distance} \(\Delta(C)\) is the minimum distance between any two distinct codewords in \(C\). 
	
	\tsf{List recovery} refers to the paradigm of decoding from \emph{soft} information, where the possibilities encompassing the soft information at any codeword location are captured by an input list. Formally, for integer parameters \(\ell,L\ge1\) and real $\rho \in [0,1)$, we say \(C\) is \tsf{\((\rho,\ell,L)\)-list recoverable} if for any sequence of input lists \(S_1,\ldots,S_n\in\binom{\Sigma}{\ell}\), we have at most \(L\) codewords in \(C\) that agree with the soft information represented by the sequence of lists \((S_1,\ldots,S_n)\) a $1-\rho$ fraction of the time, that is,
	\[
	\big|\big\{c\in C:|\{i\in[n]:c_i\not\in S_i\}|\le\rho n\big\}\big|\le L.
	\]
	
	To write this more compactly, we first extend the definition of the Hamming metric function to allow one of the arguments to consist of a tuple of subsets of $\Sigma$. That is, if $S_1,\dots,S_n \subseteq \Sigma$ and $S \coloneqq (S_1,\dots,S_n)$, for $x \in \Sigma^n$ we define
	\begin{align*}
		\Delta(x,S)&\coloneqq \frac{1}{n}|\{i \in [n]:x_i \notin S_i\}| \\
		&= \min\{\Delta(x,y):y \in S_1 \times \cdots \times S_n\}. 
	\end{align*}
	We can then define \tsf{list recovery balls}, which generalize the notion of Hamming balls. For a tuple $S = (S_1,\dots,S_n) \in \binom{\Sigma}{\ell}^n$ -- where here and throughout, for a set $X$ and integer $1 \leq k \leq |X|$, $\binom{X}{k}$ denotes the family of size-$k$ subsets of $X$ -- and a radius $\rho \in [0,1)$, define:
	\[
	B_\rho(S)\coloneqq\{x \in \Sigma^n:\Delta(x,S) \leq \rho\} \ .
	\]
	Pictorially, such objects can be viewed as a ``puffed-up combinatorial rectangle.'' That is, one starts with the combinatorial rectangle $S_1 \times S_2 \times \cdots \times S_n \subseteq \Sigma^n$, and then places a radius $\rho$ Hamming ball on each point of the combinatorial rectangle.  With this notation in place, we can now compactly write what it means for a code $C$ to be $(\rho,\ell,L)$-list recoverable:
	\begin{align}
		\forall~ S \in \binom{\Sigma}{\ell}^n, ~~|C \cap B_\rho(S)| \leq L \ . \label{eq:list-rec-corruptions-def}
	\end{align}
	See~\cref{fig:combined-puffed-rectangle} for an illustration.

	\begin{figure*}[h]
		\centering
		% --- First TikZ Picture (Transposed) ---
		\begin{tikzpicture}[scale=1.2]
			% All (x,y) coordinates have been swapped to (y,x)
			\foreach \y in {1.5, 2.0, 2.5} 
			{\draw[draw=black, thick, fill=blue!30] (\y,1.0) circle (0.7 cm);
				\draw[draw=black, thick, fill=blue!30] (\y,4.0) circle (0.7 cm);}
			
			\foreach \x in {1.0, 1.5, 2.0, 2.5, 3.0, 3.5, 4.0}
			{\draw[draw=black, thick, fill=blue!30] (1.0,\x) circle (0.7 cm);
				\draw[draw=black, thick, fill=blue!30] (3.0,\x) circle (0.7 cm);}
			
			\filldraw[draw=black, thick, fill=blue!30] (1,1) rectangle (3,4);
			
			\foreach \x in {-2,...,12}
			\foreach \y in {0,...,10}
			{\filldraw[draw=black!70,fill=black!70] (0.5*\y,0.5*\x) circle (0.06cm);}
		\end{tikzpicture}
		% --- Spacing between pictures ---
		\hspace{1.5cm}
		% --- Second TikZ Picture (Transposed) ---
		\begin{tikzpicture}[scale=1.2]
			% All (x,y) coordinates have been swapped to (y,x)
			\foreach \t in {2,3,4,5}
			{\fill[fill=blue!30] (0.5+0.5*\t,-1+\t) circle (0.7cm);}
			
			\foreach \t in {3,4}
			\foreach \x in {-0.5,0,0.5}
			\foreach \y in {-0.5,0,0.5}
			{\filldraw[draw=black, dashed, fill=blue!30] (\y+0.5+0.5*\t,\x-1+\t) circle (0.7cm);}
			
			\foreach \x in {0,0.5}
			\foreach \y in {-0.5,0,0.5}
			{\filldraw[draw=black, dashed, fill=blue!30] (\y+0.5+0.5*2,\x-1+2) circle (0.7cm);}
			
			\foreach \x in {-0.5,0}
			\foreach \y in {-0.5,0}
			{\filldraw[draw=black, dashed, fill=blue!30] (\y+0.5+0.5*5,\x-1+5) circle (0.7cm);}
			
			\draw[draw=black, thick] (1,1) rectangle (3,4);
			
			\foreach \x in {-2,...,12}
			\foreach \y in {0,...,10}
			{\filldraw[draw=black!70,fill=black!70] (0.5*\y,0.5*\x) circle (0.06cm);}
			
			\foreach \t in {0,...,7}
			{\filldraw[draw=red, fill=red] (0.5+0.5*\t,-1+\t) circle (0.1cm);}
		\end{tikzpicture}
		% --- Combined Caption and Label ---
		\caption{\textbf{Left:}~(Reproduced from~{\cite{resch2020}}) An illustration of a ``puffed-up rectangle'' $B(S,\rho)$, created by placing a ball of radius $\rho$ around each point in a combinatorial rectangle $S$.\protect\\\textbf{Right:} The same rectangle shown with a code $C$ (red vertices). The balls are filtered to show only those with a non-trivial intersection with $C$.}
		\label{fig:combined-puffed-rectangle}
	\end{figure*}
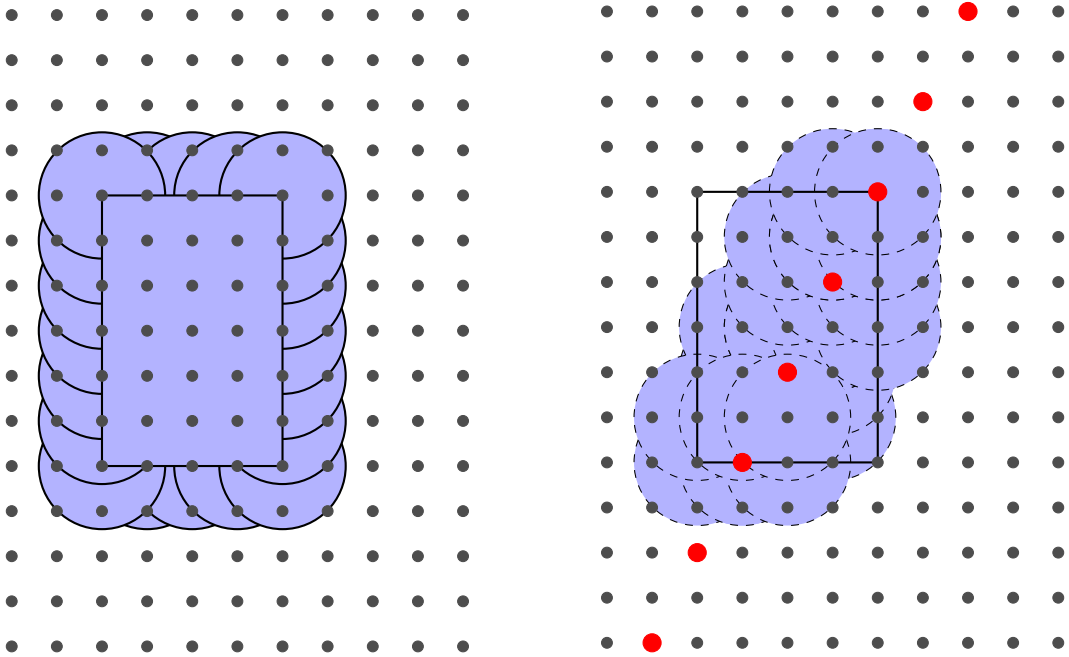
	
	\subsubsection*{Special cases} Two special cases deserve additional attention. When \(\ell=1\) one recovers \tsf{list decoding}, which is the relaxed version (allowing for more errors) of the standard \tsf{unique decoding} for error-correcting codes. In this model, one receives as input a (hard-decision) vector $y \in \Sigma^n$, and one aims to output all codewords $c$ that agree with $y$ sufficiently, in the sense that $\Delta(c,y)\leq \rho$; assuming the number of such codewords is always at most $L$ for any $y \in \Sigma^n$, the code is called $(\rho,L)$-list-decodable. Of course, if one additionally sets $L=1$ then one arrives at the standard notion of unique-decodability from a $\rho$ fraction of errors.
	
	Conversely, consider the case when $\rho=0$. If $\ell=1$, this notion is trivial: one is simply asking for the number of codewords equal to a certain target vector $y \in \Sigma^n$ to be at most $L$. However, once $\ell \geq 2$ this notion becomes nontrivial; in fact, as we will discuss below, it is quite interesting given its connection to objects of study in pseudorandomness. When $\rho=0$ we will refer to $(\ell,L)$-\tsf{zero-error list recovery}.

	\subsubsection*{List recovery from \emph{erasures}}  Up to now we have been implicitly discussing symbol \emph{corruptions} as our model of errors.  Another natural model of errors are \emph{erasures}, where some symbols are replaced by an erasure symbol $\bot \notin \Sigma$. For list recovery, one can consider the following scenario: for some $(1-\rho)$ fraction of the coordinates $i \in [n]$, we obtain soft-information $S_i \in \binom{\Sigma}{\ell}$, and for the remaining $\rho$ fraction of $i\in[n]$ we obtain \emph{no} information, which is modeled by $S_i = \Sigma$. Thus, we make the following definition: a code $C$ is $(\rho,\ell,L)$\tsf{-list recoverable from erasures} if for all $S_1,\dots,S_n \subseteq \Sigma$ such that $|S_i|\leq \ell$ for at least a $1-\rho$ fraction of $i \in [n]$, we have 
	\[
	|C \cap (S_1 \times\cdots\times S_n)| \leq L \ .
	\]
	Note this again naturally generalizes list decoding from erasures, which is the special case where $\ell=1$; and additionally if one insists that $L=1$ then one simply has unique decoding from erasures. 
	
	Having now defined list recovery from erasures, if necessary we will now write list recovery \tsf{from corruptions} to refer to the notion as defined in~\eqref{eq:list-rec-corruptions-def}; however, if just list recovery is written by default it refers to the corruptions model. Lastly, to provide some foreshadowing for the coming results, we remark that while zero error list recovery is a special case of both list recovery from corruptions and erasures, it seems to behave more similarly to list recovery from erasures, at least from the perspective of what \emph{linear} codes can and cannot achieve. 
	
	\subsubsection*{Encoding functions and a hierarchy}  
	%    We additionally remark that zero-error list recovery is naturally connected to (list-)decoding \emph{from erasures}. Given a vector $z \in (\Sigma \sqcup \{\bot\})^n$ with a $\rho$-fraction erasures (represented here by $\bot$), determining all the codewords agreeing with $z$ amounts to determining all the codewords in the combinatorial rectangle $S_1 \times \cdots \times S_n$ where $S_i = \{z_i\}$ if $z_i \in \Sigma$ and $S_i = \Sigma$ if $z_i = \bot$. Thus, \tsf{(list-)decoding from erasures} similarly amounts to ``finding all codewords in a combinatorial rectangle,'' although the combinatorial rectangle is of a different shape. It is natural to also consider \tsf{$(\rho,\ell,L)$-list recovery from erasures}: here, one is given input lists $S_1,\dots,S_n \subseteq \Sigma$ with the guarantee that for at least a $1-\rho$ fraction of $i \in [n]$ one has $|S_i| \leq \ell$, and one requires
	% \[
	% |C \cap (S_1 \times \cdots \times S_n)| \leq L
	% \]	
	Sometimes, it is also important to consider a code \(C\subseteq\Sigma^n\) with respect to how the codewords represent the underlying messages. With abuse of notation, we can equivalently define the code as an injective map \(C:\Sigma^k\to\Sigma^n\), for some \(k\le n\) (and so the rate \(R=k/n\)).  In particular, all the models mentioned so far are generalizations of the elementary problem of \tsf{interpolation}:  Given a received word \(w\in\Sigma^n\), determine the unique message \(m\in\Sigma^k\) (if any) satisfying \(C(m)=w\). 
	
	At the other extreme, all the models mentioned so far can be considered as special cases of \tsf{soft-decision decoding}. Given weight functions \(w_i:\Sigma\to[0,1]\) satisfying \(\sum_{a\in\Sigma}w_i(a)=1\) for every \(i\in[n]\), one requires that
	\[
	\bigg|\bigg\{c\in C:\sum_{i=1}^nw(c_i)>1-\rho\bigg\}\bigg|\le L \ .
	\]See~\cref{fig:LR-hierarchy} for the hierarchy among the different models.
	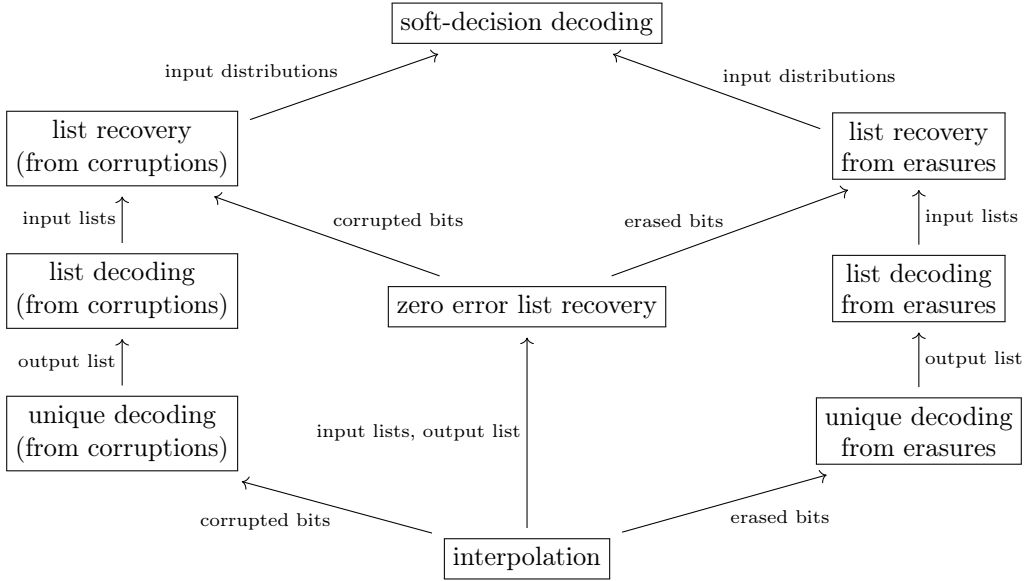
\begin{figure*}[h]
		\centering
		{\small% https://tikzcd.yichuanshen.de/#N4Igdg9gJgpgziAXAbVABwnAlgFyxMJZABgBoAmAXVJADcBDAGwFcYkQAdDgIwgA8YUYFxww+OYIyxwcAAlgBjaFjABzWQAoAZgCcIAW1lKdO5mjwE4ASgC+NkDdLpMufIRRkAjNTpNW7Ljh9JkYuXgEhETEJKRlZHRglWhgdAE9NXQMjCBMzCzBrOwcnEAxsfKIyAGYfBhY2RE4efkFhDlFxYGYwLABHVnlE5TUMvUNjU3M3QvtHZ3K3InIKWr8GpvDWqM6ALxSIWRS9HVlYuQSklNSiudKXCpQAFlJvGjr-RrCWyPboyWlzkNkmlZJlDCl6HBmAk4DcSmVXAQiM8qG81gFmhE2h0YgDBkooCp1GDDjpIdD4HD5oj3Mhlo9VvUMZsfjjgCpRDoMIx6PkqXcFkiUMtiIyPhtvti-nAIFocABaRTSNz44aqfkIh7IZ41NFMz6Yra-TrdPoDRRq0FjUnkmE3HyCVTwIigMFIZYgHAQJBkEAACxg9CgSDAzEYjD14u2EggzBwaDjpwBsxKbsQAFYaF6kM9-YHg4hQ+GaDzuDBGAAFe6LRo6LCqP04ECR9bR9lgBNyM6w4qusZIKpZ72ITw0ANBkNhiMgUvlquC9wgOsNpstjFsiZ5QSybi4Hu3NMANiH7rH+cnxZn9DLlerQpnMDlzd8+qaG5yk1EUB3e5TfYMPoniOa4GmyKidkmMj7qm-aIIOnrDrm7ytsaEgQnA267jg0H-voSCjghSAAOwgW+fzgYmhIyHW3BxtMf4gGmubZogJEvlGqHthBVE4DRdGWPYJbXnOd6Lsuja9oxsHHoRiAABxCTe840uw4mruxKFsuhmG-pJaYKbJmZXkpomqfWEmkW2sbxomZwMUeQHwch67kR2tkAnApCyNZEF2Q4lA2EAA
			\begin{tikzcd}
				& & \boxed{\text{soft-decision decoding}} & & \\
				\fbox{\shortstack{\text{list recovery} \\ \text{(from corruptions)}}} \arrow[rru, "\text{input distributions}"] & & & & \fbox{\shortstack{\text{list recovery} \\ \text{from erasures}}} \arrow[llu, "\text{input distributions}"'] \\
				\fbox{\shortstack{\text{list decoding} \\ \text{(from corruptions)}}} \arrow[u, "\text{input lists}"] & & \boxed{\text{zero error list recovery}} \arrow[llu, "\text{corrupted bits}"'] \arrow[rru, "\text{erased bits}"] & & \fbox{\shortstack{\text{list decoding} \\ \text{from erasures}}} \arrow[u, "\text{input lists}"'] \\
				\fbox{\shortstack{\text{unique decoding} \\ \text{(from corruptions)}}} \arrow[u, "\text{output list}"] & & & & \fbox{\shortstack{\text{unique decoding} \\ \text{from erasures}}} \arrow[u, "\text{output list}"'] \\
				& & \boxed{\text{interpolation}} \arrow[llu, "\text{corrupted bits}"] \arrow[rru, "\text{erased bits}"'] \arrow[uu, "{\text{input lists, output list}}"] & &
		\end{tikzcd}}
		\caption{The hierarchy between the different models.  For any edge, the upper model is more general than the lower model, and the edge label denotes the additional relaxation allowed.}\label{fig:LR-hierarchy}
	\end{figure*}
	
	\subsection{Why? Where? How?}
	
%	\paragraph*{The context.}
	
	The list recovery paradigm was initially conceived as a relaxation of list decoding, with good list recoverable codes to serve as an intermediary towards unique decoding and list decoding~\cite{GI01,GI02,GI03,GI04}. Wonderfully, this paradigm has since found widespread applications in theoretical computer science, for instance, in cryptography, construction of pseudorandom objects, group testing, and streaming algorithms. 
	% \closed{\nic{let's remove all these citations, they don't add much.}\venki{Done.}}
	
	In light of the above, obtaining a thorough understanding of list recovery has become a fundamental line of investigation in theoretical computer science. This entails studying the (tight) relationships/tradeoffs between the fundamental parameters of the code (rate, distance, alphabet size) and the list recovery parameters (input list size, output list size, list recovery radius), as well as obtaining explicit constructing codes that achieve these tradeoffs.
	
	\subsubsection*{The plan} In this article, we survey the state-of-the-art on list recoverable codes, which includes the tradeoffs between the relevant parameters, guarantees given by random constructions, and guarantees achieved by the current best explicit constructions. 
	
	An overarching story which will appear is the following. Up until recently, almost all combinatorial upper bounds for list recoverability were also algorithmic -- that is, codes were given achieving certain list recovery parameters, and the proof in fact demonstrated that one could efficiently find the output list $C \cap B_\rho(S)$. While these results are very interesting, they do not achieve optimal tradeoffs amongst all the parameters of interest. Only in recent years have direct combinatorial techniques been devised allowing for a deeper probing of what can and cannot be done in the context of list recovery. 
	
	Finally, we conclude by highlighting certain other areas in theoretical computer science where list recoverable codes play starring roles. A theme which will arise is that these application areas demand parameters which are somewhat unnatural from a coding theoretic perspective, and in certain cases necessitate novel techniques. We will make special effort to highlight interesting challenges and open problems posed by these different domains.

	\section{Warmup:  List decoding via list recovery}\label{sec:LD-via-LR}
	
	Among the earliest applications of list recovery was towards obtaining list decodable codes with better parameters.  The simplest instantiation of such an approach was via \emph{concatenated codes}.  In this section, we will see this instantiation as a warmup for the rest of the survey.  Further, we will also see the state-of-the-art improvement over concatenated codes, that are realized via the use of pseudorandom objects called \emph{expander graphs}.
	
	\subsection{List decoding concatenated codes via list recovery}\label{sec:LD-concat}
	
	%	\venki{I am using the notion of encoding map \(C:\Sigma^d\to\Sigma^n\) in the following proof.  I have mentioned it along with the hierarchy and other stuff.}
	
	To begin with, let us quickly describe the concatenated code construction.  Assume the following:
	\begin{itemize}
		\item \emph{inner code} \(C_\inn:\Sigma_\inn^k\to\Sigma_\inn^n\) having rate \(R_\inn\) (that is, \(k=R_\inn n\)) and distance \(\delta_\inn\).
		\item \emph{outer code} \(C_\out:\Sigma_\out^K\to\Sigma_\out^N\) having rate \(R_\out\) (that is, \(K=R_\out N\)) and distance \(\delta_\out\).
		\item  \(\Sigma_\out=\Sigma_\inn^k\).
	\end{itemize}
	Take any message \(m=(m_1,\ldots,m_K)\in\Sigma_\out^K\) for \(C_\out\).  This gives the corresponding outer codeword \((c_1,\ldots,c_N) \in C_\out\).  Since \(\Sigma_\out=\Sigma_\inn^k\), each \(c_t,\,t\in[N]\) is a message for \(C_\inn\).  This gives the corresponding tuple of inner codewords \((c'_1,\ldots,c'_N)\in(C_\inn)^N\).  The map \((m_1,\ldots,m_K)\mapsto(c'_1,\ldots,c'_N)\) is well-defined, and defines the \tsf{concatenated code} \(C_\out\circ C_\inn:(\Sigma_\inn)^{kK}\to(\Sigma_\inn)^{nN}\).  It follows easily~\cite[Theorem 10.1.1]{guruswami-rudra-sudan-2023-ECT} that \(C_\out\circ C_\inn\) has rate \(R_\inn R_\out\) and distance \(\delta_\inn\delta_\out\).
	
	The fundamental result that shows how list recovery can be leveraged towards list decoding is the following.  We refer the reader to~\cite{guruswami-rudra-2009-multilevel-concatenation}, for a further improved and more sophisticated construction involving \emph{multilevel concatenation}, and a decoding algorithm involving soft-decision decoding.
	\begin{proposition}[Folklore, cf.~{\cite[Lemma 2.1]{guruswami-rudra-2009-multilevel-concatenation}}]
		If \(C_\out\) is \((\xi,\ell,L)\)-list recoverable, and \(C_\inn\) is \((\rho,\ell)\)-list decodable, then \(C_\out\circ C_\inn\) is \((\xi\rho,L)\)-list decodable.
	\end{proposition}
	\begin{proof}
		Note that in this proof, we will abuse notation and denote by \(B_\rho^k\) the Hamming ball in the ambient spaces \(\Sigma_\inn^k\) as well as \(\Sigma_\out^k\) respectively.  The ambient space in each case will be clear from the context.  Also denote \(\ol{C}=C_\out\circ C_\inn\).
		
		Consider any word \(w\in\Sigma_\inn^{nN}\), and denote
		\begin{align*}
			w&=(w_1,\ldots,w_N)\\
			&=(w_{1,1},\ldots,w_{1,n},\,\ldots,\, w_{N,1},\ldots,w_{N,n}).
		\end{align*}
		We wish to show that
		\[
		\big|\ol{C}\cap B_{\xi\rho}^{nN}(w)\big|\le L.
		\]
		Since \(C_\inn\) is \((\rho,\ell)\)-list decodable, we get
		\[
		\forall\,t\in[N],\quad\big|C_\inn\cap B_\rho^n(w_t)\big|\le\ell.
		\]
		For every \(t\in[N]\), denote \(S_t=C_\inn^{-1}\big(C_\inn\cap B_\rho^n(w_t)\big)\subseteq\Sigma_\inn^k\), and notice that since \(C_\inn\) is injective, we get \(|S_t|\le\ell\).  Consider the combinatorial rectangle \(S=S_1\times\cdots\times S_N\).  Since \(C_\out\) is \((\xi,\ell,L)\)-list recoverable, we already have
		\[
		\big|C_\out\cap B_\xi^N(S)\big|\le L.
		\]
		
		We will prove our assertion by showing an injection
		\[
		\ol{C}\cap B_{\xi\rho}^{nN}(w)\lhook\joinrel\longrightarrow C_\out\cap B_\xi^N(S).
		\]
		Indeed, if we have such an injection, then it would immediately imply
		\[
		\big|\ol{C}\cap B_{\xi\rho}^{nN}(w)\big|\le\big|C_\out\cap B_\xi^N(S)\big|\le L,
		\]
		and so we can conclude that \(\ol{C}\) is \((\xi\rho,L)\)-list decodable.
		
		Now let \(c'\in \ol{C}\cap B_{\xi\rho}^{nN}(w)\), and denote
		\begin{align*}
			c'&=(c'_1,\ldots,c'_N)\\
			&=(c'_{1,1},\ldots,c'_{1,n},\,\ldots,\, c'_{N,1},\ldots,c'_{N,n}).
		\end{align*}
		So we have
		\begin{align}
			&\sum_{\substack{t\in[N]\\i\in[n]}}\mbf{1}(c'_{t,i}=w_{t,i})>(1-\xi\rho)nN.\label{eq:concat-dist}
		\end{align}
		Informally, this means \(c'\) has large agreement with \(w\).  We will now argue that this implies \(c'_t\) has large agreement with \(w_t\), for a large number of \(t\in[N]\).  Let
		\[
		T=\bigg\{t\in[N]:\sum_{i\in[n]}\tbf{1}(c'_{t,i}=w_{t,i})>(1-\rho)n\bigg\}.
		\]
		If \(|T|\le(1-\xi)N\), then we have
		\begin{align*}
			&\qquad\sum_{\substack{t\in[N]\\i\in[n]}}\mbf{1}(c'_{t,i}=w_{t,i})\\
			&=\sum_{\substack{t\in T\\i\in[n]}}\mbf{1}(c'_{t,i}=w_{t,i})+\sum_{\substack{t\in[N]\setminus T\\i\in[n]}}\mbf{1}(c'_{t,i}=w_{t,i})\\
			&\le|T|n+(N-|T|)(1-\rho)n\\
			&=|T|\rho n+(1-\rho)nN\\
			&\le((1-\xi)\rho+(1-\rho))nN\\
			&=(1-\xi\rho)nN,
		\end{align*}
		which contradicts~(\ref{eq:concat-dist}).  So \(|T|>(1-\xi)N\).  In other words, we have
		\[
		\big|\{t\in[N]:c'_t\in C_\inn\cap B_\rho^n(w_t)\}\big|>(1-\xi)N,
		\]
		which implies
		\[
		(C_\inn^{-1}(c'_1),\ldots,C_\inn^{-1}(c'_N))\in C_\out\cap B_\xi^N(S).
		\]
		Since \(C_\inn\) is injective, we get the desired injection
		\begin{align*}
			\ol{C}\cap B_{\xi\rho}^{nN}(w)&\lhook\joinrel\longrightarrow C_\out\cap B_\xi^N(S)\\
			c'&\longmapsto(C_\inn^{-1}(c'_1),\ldots,C_\inn^{-1}(c'_N)).\qedhere
		\end{align*}
	\end{proof}
	
	\subsection{Concatenation \(+\) Expansion \(=\) Magic!}\label{sec:AEL}

	An unavoidable drawback of concatenated codes is their rate v/s distance tradeoff.  To recall, the Singleton bound~\cite{komamiya-1953-MDS,joshi-1958-MDS,singleton-1964-MDS} states that a rate \(R\) code has distance \(\delta\le 1-R\).\footnote{Strictly speaking, the Singleton bound states that \(\delta\le 1-R+(1/n)\), where \(n\) is the length of the code.  Since we are only concerned with asymptotic results as \(n\to\infty\), we ignore the additive \(1/n\) term.}  The codes which attain this bound are called \tsf{maximum distance separable (MDS) codes}.  For now, we assume that \emph{MDS codes exist for all the parameters that we wish}.  We will elaborate on some subtleties at the end of this section.
	
	Continuing with the notations from~\cref{sec:LD-concat}, as mentioned before, the code \(C_\out\circ C_\inn\) has rate \(R_\inn R_\out\) and distance \(\delta_\inn\delta_\out\).  This means the concatenated code can be very far from being MDS even if the inner and outer codes are both MDS.  For instance, if we choose both \(C_\inn\) and \(C_\out\) to be MDS with \(R_\inn=R_\out\eqqcolon R\), then \(C_\out\circ C_\inn\) has rate \(R^2\) and distance \((1-R)^2\).  However, by the Singleton bound, a code with rate \(R^2\) can have distance as much as
	\[
	1-R^2>(1-R)^2.
	\]
 This is one of the many typical situations where the rich area of \emph{pseudorandomness} steps in.  One of the overarching strategies that this area employs is the following: take a \emph{best, small} object (usually obtained by brute-force search) and \emph{lift} it using an explicit \emph{good, large} object to get an explicit \emph{better (or nearly best), large} object.  A breakthrough work by~\cite{alon-edmonds-luby-1995-AEL} showed how this can be done to \emph{amplify} the distance of a concatenated code to get \emph{nearly MDS} codes.  We will quickly review this beautiful construction here, which is popularly called the \emph{AEL code}.  See~\cref{fig:AEL} for a block diagram showing a typical pseudorandom construction, as well as the AEL code construction.
	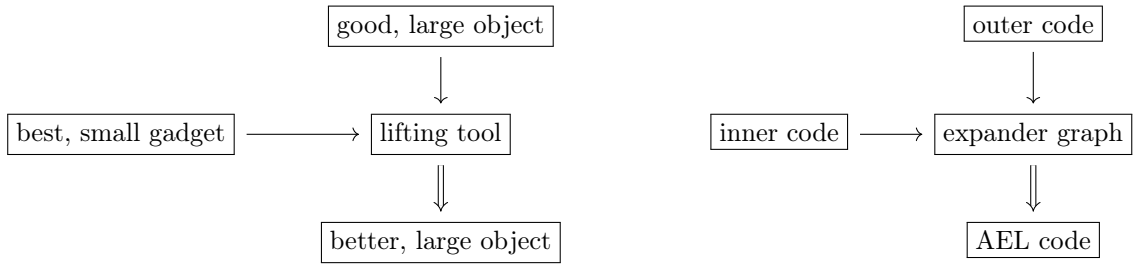
\begin{figure*}[h]
		\centering
		{\small% https://tikzcd.yichuanshen.de/#N4Igdg9gJgpgziAXAbVABwnAlgFyxMJZARgBoAGAXVJADcBDAGwFcYkQAdDgIwgA8YUYFxww+OYAHMI0UgAJG9AE6SYciNwBWMAMY4AvvpD7S6TLnyEU5UsWp0mrdl14ChIsRO7wc8uAFsmRjlJeihVAyMTM2w8AiIyOxoGFjZETh5+QWEOUXFgRiwAMzwwSTkcGUZDY1MQDFjLBNIAJnsUp3SXLPdcz2BvHFEleUUVNQ1tPRro+vM4q2QAFgp2xzSM12yPfIhmYbkdaBgZuoaLeJQAZls11OdMtxy8iSwwMBglQ+PTmIvFlZJBz3LqPbZ9fJiND0MCwL6SJT0NAAC1+c0al2WrTunU2PWe-QAggBRAAy31gM3sglUCBQoCKSgg-iQNhAlSQLVmjOZSDI7IgnO5TJZiBaNA5iBuIDgyOKOFZNG8sMVIGRMDCSDAzEYjBohQ+7Cge24jDYEvoWEY7EghuFvMQKwFSAAbPbRQBWCWCxBuuo80Uu71IADs+reG2NzFN5rVGqgWp1evZlut6VtbH0lH0QA
			\begin{tikzcd}
				& {\boxed{\text{good, large object}}} \arrow[d]     &  &                                     & \boxed{\text{outer code}} \arrow[d]                 \\
				{\boxed{\text{best, small gadget}}} \arrow[r] & \boxed{\text{lifting tool}} \arrow[d, Rightarrow] &  & \boxed{\text{inner code}} \arrow[r] & \boxed{\text{expander graph}} \arrow[d, Rightarrow] \\
				& {\boxed{\text{better, large object}}}             &  &                                     & \boxed{\text{AEL code}}                            
			\end{tikzcd}
		}
		\caption{A typical pseudorandom construction strategy (left), and the AEL code construction (right).}\label{fig:AEL}
	\end{figure*}
	
	An \tsf{\((N,D,\lambda)\)-biregular bipartite expander} is a \((D,D)\)-biregular bipartite graph \(G=(U,V,E)\) with \(|U|=|V|=n\), such that the adjecancy matrix has second eigenvalue \(\sigma_2(G)\le\lambda D\).\footnote{Since \(G\) is \((D,D)\)-biregular, the largest eigenvalue \(\sigma_1(G)=D\).  There are other notions of expansion like vertex expansion, edge expansion, etc. which we do not consider here.  However, in principle, these are all equivalent to spectral expansion.  Furthermore, expander graphs need not be biregular or bipartite in generhere we considerider the biregular bipartite structure for simplicity.}  Now assume the following:
	\begin{itemize}
		\item  \emph{inner code} \(C_\inn:\Sigma_\inn^k\to\Sigma_\inn^D\) having rate \(R_\inn\) (that is, \(k=R_\inn D\)) and distance \(\delta_\inn\).
		\item  \emph{outer code} \(C_\out:\Sigma_\out^K\to\Sigma_\out^N\) having rate \(R_\out\) (that is, \(K=R_\out N\)) and distance \(\delta_\out\).
		\item  \(\Sigma_\out=\Sigma_\inn^k\).
		\item  \(G=([N],[N],E)\), an \((N,D,\lambda)\)-bipartite expander.
	\end{itemize}
	Order the \(D\) edges incident on each vertex arbitrarily.  For each edge \(e\in E\), define \(\mtt{L}(e)=(\ell,i)\in[N]\times[D]\) and \(\mtt{R}(e)=(r,j)\in[N]\times[D]\), if \(e\) is the \(i\)-th edge incident on the left vertex \(\ell\), and the \(j\)-th edge incident on the right vertex \(r\).
	
	The codewords in the \tsf{AEL code} \(C_\AEL\) are defined via a one-to-one correspondence (using the expander \(G\)) with the codewords of \(C_\out\circ C_\inn\).  Consider any codeword in \(C_\out\circ C_\inn\),
	\[
	c=(c_{1,1},\ldots,c_{1,D},\,\ldots,\,c_{N,1},\ldots,c_{N,D}).
	\]
	The corresponding codeword in \(C_\AEL\) is
	\[
	c'=(c'_{1,1},\ldots,c'_{1,D},\,\ldots,\,c'_{N,1},\ldots,c'_{N,D}),
	\]
	which is a permutation of the entries of \(c\), defined by the relation
	\begin{align*}
		&\quad c_{\ell,i}=c'_{r,j}\\
		\iff\n&\exists\,e\in E,\n \mtt{L}(e)=(\ell,i),\,\mtt{R}(e)=(r,j).
	\end{align*}
	Informally, the symbols of \(c\in C_\out\circ C_\inn\) are redistributed as per the edges of the graph \(G\) to obtain \(c'\in C_\AEL\).  Further, \(c'\) is now considered as a vector of length \(N\) over the alphabet \(\Sigma_\inn^D\).  This simple redistribution via the expander graph as a \emph{lifting tool} leads to the following `magical' distance amplification.
	\begin{theorem}[{\cite{alon-edmonds-luby-1995-AEL}}]\label{thm:AEL}
		The code \(C_\AEL\) has rate \(R_\inn R_\out\), and distance at least
		\[
		\delta_\inn-\frac{\lambda}{\delta_\out}.
		\]
	\end{theorem}
	
	We can instantiate the above construction as follows.  Fix any constant rate \(R\in(0,1)\), and arbitrarily small constant \(\epsilon\in(0,1-R)\).  Let \(\lambda=\epsilon^2\) and \(D=4/\lambda^2=4/\epsilon^4\).  By a classic expander graph construction~\cite{lubotzky-phillips-sarnak-1988-ramanujan} using \emph{Cayley graphs}, we then have an explicit \((N,D,\lambda)\)-bipartite expander \(G\).\footnote{The relation \(D=4/\lambda^2\) that we choose is essentially optimal for spectral expanders.  Such optimal spectral expanders are called \tsf{Ramanujan graphs}.}  Take MDS codes \(C_\inn\) and \(C_\out\) with
	\begin{align*}
		R_\inn&=R,&\delta_\inn&=1-R,\\
		R_\out&=1-\epsilon,&\delta_\out&=\epsilon.
	\end{align*}
	Then~\cref{thm:AEL} gives the explicit code \(C_\AEL\) with rate at least \(R-\epsilon\), and distance at least \(1-R-\epsilon\).
	
	We conclude this discussion by mentioning that among the most recent breakthroughs, we now know that suitably instantiated AEL codes can be algorithmically list decoded~\cite{jeronimo-mittal-srivastava-tulsiani-2025-AEL}, and list recovered~\cite{srivastava-tulsiani-2025-AEL} up to capacity.  A crucial ingredient of these algorithms is the \emph{Sum of Squares (SoS) method}, which adapts a classical \emph{proofs-to-algorithms} framework~\cite{fleming-kothari-pitassi-2019-proofs-to-algorithms} to the setting of decoding from errors.
	
	\subsubsection*{Curioser and curioser...}  We end this section with some subtleties that we conveniently avoided earlier.  These are interesting lines of investigation in their own right, but beyond our current scope.
	\begin{enumerate}[(a)]
		\item  \emph{MDS codes are not always possible!}\quad  The above situation is idealized in the sense that we assume we can have MDS inner and outer codes.  In fact, the best tradeoff that we can hope for is tied to the alphabet size.  With code concatenation, the inner alphabet is usually small, say binary. So hoping for the inner code to be MDS isn't the usual setting.  What we typically have is inner code attaining the \emph{Gilbert-Varshamov (GV) bound} and outer code MDS, and then the concatenated code hits what is called the \emph{Zyablov bound}.  We will not delve into further details.  Our idealized setting above is essentially when we take \(|\Sigma_\inn|\ge2^{\poly(1/\epsilon)}\), which allows us to have inner codes of rate \(R_\inn\) and distance \(\delta_\inn\ge1-R_\inn-\epsilon\).
		
		\item\emph{Optimal list decoding is not always possible!}\quad  In the setting of list decoding, one also obtains a suboptimal tradeoff between the rate and decoding radius. Specifically, such constructions shine when the decoding radius is set to $1-1/|\Sigma_\inn|-\eps$, and then one sets the rate $R = \poly(\eps)$; whereas one would hope for $R=\Omega(\eps^2)$ when $\Sigma_\inn = O(1)$, the rate in such (explicit)\footnote{We do remark that Guruswami and Rudra~\cite{GR08} showed that there \emph{exist} concatenated codes list decodable up to capacity. But an explicit construction is not yet known.} constructions is unfortunately worse.
	\end{enumerate}
	
	\section{What do random codes promise us?}
	
	In coding theory (and in combinatorics more broadly) one often would like to determine optimal parameters for certain structures. In the case of list recoverable codes, our dream is to understand the optimal tradeoff between a code's rate (parametrized by $R$) and its list recoverability (parametrized by $\rho$, $\ell$ and $L$; additionally by the noise model of corruptions or errors). In this section we will provide a (somewhat coarse) answer to this question. 
	
	As a lens on this problem (which, in its most general form, appears \emph{very} difficult) we will look to ascertain what happens for ``most'' codes. Namely, we will consider sampling a code $C$ of a prescribed rate according to some distribution, and then determine what sort of list recoverability this code is likely to satisfy. 
	
	Firstly, we consider what we call \tsf{plain random codes (PRCs)} or rate $R$, which are subsets $C \subseteq \Sigma^n$ obtained by including each $x \in \Sigma^n$ in $C$ with probability $q^{-(1-R)n}$ (recall $q=|\Sigma|$), making these decisions \emph{independently} for each $x$.\footnote{Note that such a code $C$ has size $q^{Rn}$ in expectation, and by a Chernoff bound can be shown to have size $\geq q^{Rn}/2$ except with exponentially small probability. One could also choose a uniformly random subset of size $q^{Rn}$ which behaves similarly for all intents and purposes, but the given model makes for easier analysis.}
	
	If we additionally insist $q$ be a prime power and identify $\Sigma = \F_q$, a finite field with $q$ elements, then we can talk of \tsf{random linear codes (RLCs)}, obtained by sampling $H \in \F_q^{n-k \times n}$ for $k=Rn$ uniformly at random and defining $C := \{x \in \F_q^n : Hx=0\}$.\footnote{Other reasonable models include choosing a uniformly random subspace of dimension $k$, or choosing $G \in \F_q^{n \times k}$ uniformly at random and defining $C := \{Gx : x \in \F_q^{k}\}$. Conditioned on $H$ and $G$ having full rank (which, for constant $R$, occurs with probability $1-q^{-\Omega(n)}$), these models are identical.} This way we ensure we sample a more \emph{structured} code, a topic we discuss further later. Additionally, if $q \geq n$ we can consider \tsf{random Reed-Solomon (RRS)} codes, obtained by sampling $\alpha_1,\dots,\alpha_n \in \F_q$ uniformly subject to being distinct,\footnote{Once $q \gg n^2$ one could sample the $\alpha_i$'s independently, and still be certain that with high probability the $\alpha_i$'s will be distinct.} and defining $C := \{(f(\alpha_1),\dots,f(\alpha_n)):f \in \F_q[X],\deg f \leq k-1\}$.  
	
	To set the stage we will present the \emph{list recovery capacity theorem}, which gives an overarching tradeoff between the parameters of list recoverable codes. This essentially sets the \emph{rules of the game} by specifying what performance can be achieved by PRCs.  We then know clearly what goals to aim towards, when it comes to the quest for more structured (or, more ambitiously, \emph{explicit}) codes.
	
	\subsection{List recovery capacity}
	
	Let \(|\Sigma|=q\), and for \(\ell\in[1,q-1]\), let \(h_{q,\ell}\) denote the \tsf{\((q,\ell)\)-entropy function}, that is,
	\[
	h_{q,\ell}(x)\coloneqq x\log_q\bigg(\frac{q-\ell}{x}\bigg)+(1-x)\log_q\bigg(\frac{\ell}{1-x}\bigg)
	\]
	for \(x\in(0,1-\ell/q)\).
	\begin{theorem}[List recovery capacity theorem, {cf.~\cite{resch2020}}]\label{thm:LR-capacity-theorem}
		Consider any \(q\ge2,\,\ell\in[1,q-1]\), and \(\rho\in(0,1-\ell/q),\,\epsilon\in(0,1-h_{q,\ell}(\rho))\).  For all sufficiently large \(n\),
		\begin{itemize}
			\item  {\normalfont\tsf{Possibility.}}\quad  a random code of length \(n\) having rate \(1-h_{q,\ell}(\rho)-\epsilon\) is \((\rho,\ell,O(\ell/\epsilon))\)-list recoverable with probability at least \(1-q^{-\epsilon n}\).
			\item  {\normalfont\tsf{Impossibility.}}\quad  any code of length \(n\) having rate \(1-h_{q,\ell}(\rho)+\epsilon\) is not \((\rho,\ell,q^{o(\eps n)})\)-list recoverable.
		\end{itemize}
	\end{theorem}
	Since~\cref{thm:LR-capacity-theorem} has precisely specified the limits of possibility, \emph{your mission, should you choose to accept}, is the following.
	\begin{question}\label{que:capacity-theorem}
		Find explicit codes that achieve the parameters guaranteed by random codes as in~\cref{thm:LR-capacity-theorem}. (There are constructions achieving similar parameters in certain regimes, but not all regimes.)
		%In particular, find \tbf{explicit codes achieving list recovery capacity}.
	\end{question}
	
	%Write a quick proof of capacity theorem, following Nic's thesis, and/or Mary's course notes -- main steps in the proof, but skip tedious computations, if any.  Preferably, stress how this is an example of the probabilistic method.
	
	In the remainder of this section, let us provide the by-now standard proof of the capacity theorem. 
	% The proof is a natural adaptation of Elias's~\cite{elias-1957-list-decoding} argument which was provided for list decoding. \nic{check this paper, make sure it does what I say it does...}
	\begin{proof}[Proof of~{\cref{thm:LR-capacity-theorem}}]
		To begin with, let us clarify that in our argument, by ``sufficiently large \(n\)'', we mean \(n\ge 9/\epsilon^2\).  In the possibility result, it is enough to take output list size \(L=\lceil\ell/\eps\rceil\).  In the impossibility result, in fact, we show that list recoverability is not possible even with output list size \(q^{\eps n/2}\).
		
		The capacity theorem may appear quite obscure at first blush -- what is this $(q,\ell)$-entropy function? -- but it is quite natural once one is made aware of the following estimate. For a tuple $S = (S_1,\dots,S_n) \in \binom{\Sigma}{\ell}^n$, the radius $\rho$ list-recovery ball centered at $S$ has size
		\begin{align}
			\frac{q^{h_{q,\ell}(\rho)\cdot n}}{\sqrt{2\pi n\rho(1-\rho)}} \leq |B_\rho(S)| \leq q^{h_{q,\ell}(\rho)\cdot n} \ . \label{eq:list-rec-ball-estimate}
		\end{align}
		From here, both the possibility and the impossibility results from the capacity theorem are fairly easy to establish. Both follow from the probabilistic method. For the possibility side, consider a PRC $C$ of rate $R$. $C$ fails to be $(\rho,\ell,L)$-list recoverable if there is an input list tuple $(S_1,\dots,S_n) \in \binom{\Sigma}{\ell}^n$ for which $B_\rho(S)$ contains $L+1$ codewords. By union bounding over the at most $q^{\ell n}$ choices for input list tuples and the at most $q^{(L+1) \cdot h_{q,\ell}(\rho) \cdot n}$ many $(L+1)$-sized subsets of a given radius $\rho$ list-recovery ball (this uses the upper bound from \eqref{eq:list-rec-ball-estimate}), one finds that so long as $R \leq 1-h_{q,\ell}(\rho) - \eps$ the probability the code fails to be $(\rho,\ell,L)$-list-recoverable is at most $q^{-\eps n}$. Crucially, we use here that the events $x \in C$ are \emph{independent} for different $x \in \Sigma^n$, so for any $(L+1)$-sized $X \subseteq \Sigma^n$, $\Pr[X \subseteq C] = q^{(R-1)n(L+1)}$.
		
		The impossibility side is established via a standard averaging argument. Consider a uniformly random choice for $(S_1,\dots,S_n) \in \binom{\Sigma}{\ell}^n$; the expectation $\mb E|C \cap B_S(\rho)|$ can be lower bounded by $|C| \cdot |B(S_0,\rho)|/q^n$, where $S_0 \in \binom{\Sigma}{\ell}^n$ is any choice of input list tuple (note that the size of a list recovery ball is independent of the center). By the lower bound from \eqref{eq:list-rec-ball-estimate}, when $R \geq 1-h_{q,\ell}(\rho)+\eps$ one finds the expectation is $\geq q^{\eps n}/3\sqrt{n}\ge q^{\epsilon n/2}$. 
	\end{proof}
	
	For convenience, we now introduce the notation $\Rcor(\rho,\ell)$ for the capacity $1-h_{q,\ell}(\rho)$. 
	
	For list recovery \emph{from erasures}, one can prove a capacity theorem via a completely analogous argument: for $(\rho,\ell,L)$-list recovery from erasures, the capacity is given by
	\begin{align}
		\Rera(\rho,\ell):= 1-\rho-(1-\rho)\log_q\ell \ . \label{eq:erasures-capacity}
	\end{align}
	Here, rather than \eqref{eq:list-rec-ball-estimate}, one uses the following simple identity: if $S_1,\dots,S_n \subseteq \Sigma$ which have size $\ell$ for $(1-\rho)$ fraction of the $i$, and otherwise have size $q$, then
	\[
	|S_1 \times \cdots \times S_n| = \ell^{(1-\rho)n} \cdot q^{\rho n} = q^{(\rho + (1-\rho)\log_q\ell)n}
	\]
	
	In both cases, we abbreviate $\Rcor$ and $\Rera$ when $\rho,\ell$ are clear from context. 
	
	\subsection{Meditation on the capacity theorem}
	
	Here, we list a few more observations around the capacity theorem.
	
	\paragraph*{The entropy function.}  An important fact about the $(q,\ell)$-entropy function is that on the interval $(0,1-\ell/q)$ the function increases from $0$ to $1$. Thus, so long as $\rho < 1-\ell/q$ there exist positive rate $(\rho,\ell,L)$-list recoverable codes. For $\rho > 1-\ell/q$ such codes do not exist. See \Cref{fig:ql-entropy}.
	
	\begin{figure*}[h]
		\centering
		\includegraphics[scale=0.7]{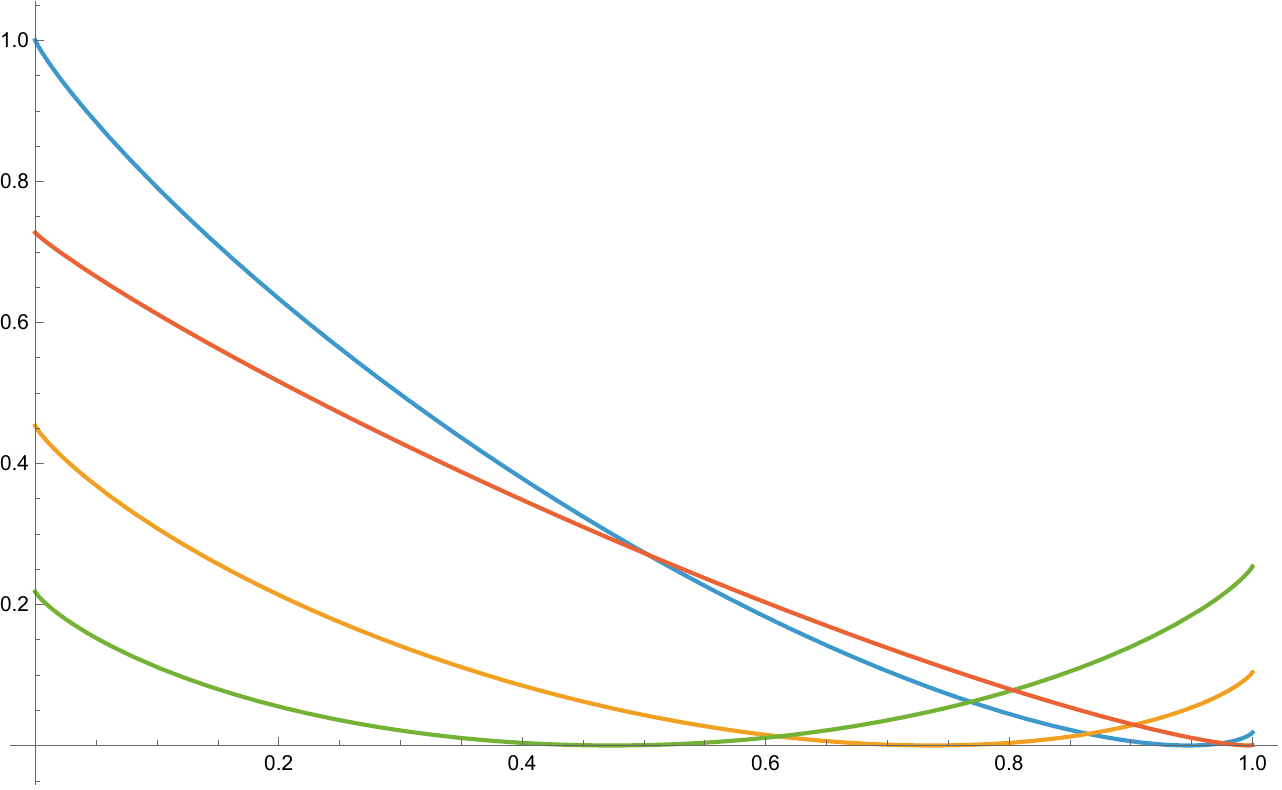}
		\caption{Plots of $1-h_{q,\ell}(x)$ for various values of alphabet $q$ and input list size $\ell$. In blue, $q=19$ and $\ell=1$; in orange, $q=19$ and $\ell=5$; in green, $q=19$ and $\ell=10$; and in red, $q=2048$ and $\ell=8$. Observe that $1-h_{q,\ell}(0) = 1-\log_q\ell$ and $1-h_{q,\ell}(1-\ell/q)=0$, and that $1-h_{q,\ell}$ decreases monotonically between these endpoints. Observe further that when $q$ is very large compared to $\ell$ (cf. the red line), one obtains essentially a straight line with $y$-intercept $1-\log_q\ell$ and $x$-intercept $1-\ell/q$.}\label{fig:ql-entropy}
	\end{figure*}
	
	\paragraph*{Achieving capacity and the Elias bound.} Based on \Cref{thm:LR-capacity-theorem}, we will introduce the following terminology: if a code of rate $R = 1-h_{q,\ell}(\rho)-\eps$ is $(\rho,\ell,L)$-list recoverable with $L \leq \poly(n)$ (viewing all other parameters $\rho,\ell,q,\eps$ as constants, where $\eps>0$ can be taken arbitrarily small), then we will say the code \tsf{achieves list recovery capacity}. If additionally the output list size $L$ can be taken as $O(\ell/\eps)$ (which was the case for PRC's), then we will say the code \tsf{achieves the Elias bound}.
	
	\paragraph*{Two disjoint regimes for alphabet size.}  As we shall see, all the three fronts of list recovery (the good, the bad, and the unknown) have significantly different forms depending on whether the alphabet is `large' or `small'.  This disparity can be traced back to the fundamentally different behavior of the \((q,\ell)\)-entropy function in the cases when \(q\) is `large' or `small'.  Let us make this precise.  Separating a linear term from \(h_{q,\ell}(x)\), we get
	\begin{align*}
		h_{q,\ell}(x)&=x+x\log_q\bigg(\frac{1-(\ell/q)}{x}\bigg)\\
		&\quad+(1-x)\log_q\bigg(\frac{\ell}{1-x}\bigg).
	\end{align*}
	It can be checked that for any \(\epsilon\in(0,1/2)\), if we set \(q\ge(\ell+1)^{1/\epsilon}\), then we have
	\[
	|h_{q,\ell}(x)-x|\le\epsilon\cdot\frac{\ln(2(\ell+1))}{\ln(\ell+1)}\le2\epsilon
	\]
	for all \(x\in(0,1-\ell/q)\).  Note that in this case, the chosen \(\epsilon\) is already within the range assumed in~\cref{thm:LR-capacity-theorem}.  Therefore, the possibility result corresponds to rate at most \(1-\rho-3\eps\), and the impossibility result corresponds to rate at least \(1-\rho-2\eps\).  In light of this, the common consensus is that the \tsf{large alphabet regime} is when \(q\ge\ell^{\,\Omega(1/\epsilon)}\), and \(h_{q,\ell}(x)\) is uniformly \(O(\epsilon)\)-close to \(x\), and the complementary case is the \tsf{small alphabet regime}. 
	
	\emph{We are now ready for launch... T minus two, T minus one, and liftoff!}
	
	%	\venki{The method of types and LCL are relevant for upper bounds.  This discussion is the main focus of upper bounds, and so there is now a section for it. The other upper bounds appear later.}
	
	\section{List recovery of structured codes: the good, and the bad}
	
	%	Recall that the important parameters to be understood are input list size, output list size, alphabet size, rate, distance, and list recovery radius.
	
	%	We will survey the recent results in this direction, including the results for random linear codes and semi-explicit families of linear codes~\cite{LMS24,LS25}, explicit families of linear codes~\cite{kopparty-ron-zewi-saraf-wootters-2023-list-decoding,tamo-2023-tighter-list-size,chen-zhang-2024-FRS-list-size}, as well as semi-explicit families of nonlinear codes~\cite{KM25}.
	
	By~\cref{thm:LR-capacity-theorem}, we know that for any \(\rho\in(0,1-\ell/q)\), a PRC having rate at most \(\Rcor-\eps\) is \((\rho,\ell,O(\ell/\eps))\)-list recoverable with high probability, and an analogous result holds for list recovery from erasures. In particular, \emph{some code} exists with these parameters. However, plain random codes have no structure that one could hope to exploit -- indeed, even describing succinctly such an exponentially large object is in general impossible, let alone efficient encoding or decoding. 
	
	In light of this, researchers have analyzed whether or not \emph{structured} codes can achieve the same list recoverability as PRCs. The most popular ``structure'' that researchers have investigated is linearity: namely, one considers the behaviour of a random \emph{linear} code. As motivation, many code operations (such as concatenation or tensoring) and applications (as we survey in \Cref{sec:beyond-coding}) require linear codes, and additionally in some of these scenarios simply an existential result demonstrating list recoverable linear codes exist is sufficient. In this section we will survey what is known: in certain cases, we will see a price to pay for linearity; in other cases, this price is significantly reduced (potentially even to zero). 
	
	\subsection{Linear codes: Good? Bad? A bit of both?}\label{sec:ZP}
	
	% \venki{Perhaps some other catchphrase?  We have already `set the stage' with the capacity theorem.  In fact, we have completed `liftoff'.  So we are flying now :-)}
	
	% \venki{Suggestion -- Linear codes: Good? Bad? Or a bit of both?}
	
	%    \nic{great :) dropped the "or" for conciseness.}
	
	To begin this discussion, it is worth sharing an argument due to Zyablov and Pinsker~\cite{zyablov-pinsker-1981}, which was originally provided for analyzing the list-decodability of RLCs but generalizes easily to list-recovery as well. Recall that for a PRC $C$ of rate $R$, we pointed out that for any $(L+1)$-sized subset $X$, we had $\Pr[X \subseteq C] = q^{(R-1)n(L+1)}$. For an RLC this is no longer the case: for example, observe that if $x,y \in C$ then necessarily $x+y \in C$, and additionally $\lambda x \in C$ for any $\lambda \in \F_q$. 
	
	But all is not lost: note that if a set $X \subseteq \F_q^n$ is \emph{linearly independent}, then we indeed have $\Pr[X \subseteq C] = q^{-(R-1)n|X|}$. Zyablov and Pinsker now exploit the fact that in any subset $X \subseteq \F_q^n$ of size $L+1$, at least $\log_q(L+1)$ of them are linearly independent. Thus, we can say (pessimistically) that if $X$ is a set of size $L+1$, then $\Pr[X \subseteq C] \leq q^{-(1-R)n\log_q(L+1)}$. By making $L+1$ exponentially larger -- namely, $q^{\ell/\eps}$ -- than we needed for PRCs, we can conclude that the RLC is with high probability $(\rho,\ell,L)$-list recoverable.\footnote{Again, the same idea applies to list recovery from erasures.}
	
	Now, this argument already establishes something quite positive: we can say that RLC's \emph{achieve list recovery capacity}. However, we cannot make the stronger conclusion that they achieve the \emph{Elias bound}. We will now survey some results below which attempt to determine when an RLC can and cannot have smaller output list size $L$. 
	
	But before diving into the precise results, we introduce a recent technique that allows for a deeper probing of combinatorial properties of RLCs. 
	
	% In this section, we explore the state-of-the-art explicit constructions towards this goal. We begin by jumping straight into the most recent developments in this direction. \nic{I will update this paragraph -- we're more surveying ``structured'' codes.}
	
	\subsection{The new player in town: a method of \emph{forbidden structures}}

	A basic heuristic that drives several lines of investigation is that desired properties can be characterized by the absence of certain \emph{forbidden structures}. Pursuing this heuristic has led to some amazing technical developments like the theory of \emph{forbidden minors} in matroid theory~\cite{oxley-1989-excluded-minor,geelen-2000-excluded-minors,mayhew-2012-rota}, the recently emerging \emph{container method}~\cite{balogh-morris-samotij-2018-hypergraph-containers,balogh-samotij-2020-container-lemma,zamir-2023-containers}, Ramsey theory~\cite{conlon-fox-sudakov-2015-ramsey,robertson-2021-ramsey}, the study of forbidden subgraphs in random graphs~\cite{babai-simonovits-spencer-1990-random-graphs,balogh-bollobas-simonovits-2004-forbidden-subgraphs,haviv-2019-forbidden-subgraphs}, etc.  Here, we shall see another concrete variant of this heuristic tailored towards capturing forbidden structures defined by local constraints in linear codes.  In fact, depending on whether the alphabet size is large or small, there are two different realizations.
	
	\subsubsection*{The method of Local Coordinatewise Linear (LCL) properties (over large alphabet)}
	
	We begin with the large field case, and explain how local coordinate-wise linear (LCL) properties capture list recoverability. In fact, this tool applies to a broader class of \emph{local} properties.\footnote{This terminology may cause confusion with topics like locally-testable codes, where one only queries a few points from a codeword before making a decision/performing a computation. Such topics are in fact not \emph{local} properties in the sense that we use here.} Here is the idea: note that to demonstrate that a code is list recoverable (either from corruptions or erasures), one must show that whenever one is given a ``bad'' set of $L+1$ vectors (say, all lying in a radius $\rho$ list recovery ball) it is \emph{not} the case that they are all codewords. 
	
	Now, take this ``bad'' $(L+1)$-sized set, and make a $(L+1)\times n$ matrix $A$ from it. Note that the coordinates of $A$ must ``agree a lot,'' at least if $\ell \ll L+1$ and $\ell \ll q$. This is easiest to see in the case of zero-error list recovery: whereas we could have up to $L+1$ different values in per column of $A$, in fact there must be at most $\ell$ in such a case. Looking at each column of $A$, we can in fact consider a partition of $[L+1]$ into $\ell$ parts, where two coordinates are in the same part if they take on the same value. The vectors $v \in \F_q^{L+1}$ satisfying these constraints in fact form a linear space. That is, for each $i \in [n]$, we can consider a subspace $V_i \leq \F_q^{L+1}$; the $n$-tuple $\mc V = (V_1,\dots,V_n)$ is then called a \emph{linear profile}.\footnote{These ideas generalize readily to list recovery from corruptions or erasures.} 
	
	From the above discussion we see that list recovery is defined by forbidding some family $\mathcal F$ of local profiles $\mc V$, where we say that a code $C$ contains a local profile $\mc V = (V_1,\dots,V_n)$ if it contains the rows of a matrix $A \in \F_q^{(L+1) \times n}$, where for each $i \in [n]$ the $i$-th column of $A$ lies in $V_i$. Any property defined by forbidding such a family of local profiles is called \tsf{locally coordinate-wise linear (LCL)}. %Crucially, one can bound the number of such profiles (crudely) by $\ell^{(L+1)n}$.
	
	The important fact Levi, Mosheiff and Shagrithaya~\cite{LMS24} establish is that all LCL properties undergo a ``threshold phenomenon.'' Informally speaking, this means that given any LCL property $\mc P$ (e.g., $(\rho,\ell,L)$-list recovery) with family of local profiles $\mathcal F$, there exists a rate $R_{\mc P} \in [0,1]$ such that: 
	\begin{itemize}
		\item if $R < R_{\mc P}$ then a random linear code of rate $R$ will satisfy $\mc{P}$ with probability $1-|\mc F| \cdot q^{-\Omega(n)}$;
		\item if $R > R_{\mc P}$ then a random linear code of rate $R$ will \emph{not} satisfy $\mc{P}$ with probability $1-q^{-\Omega(n)}$. 
	\end{itemize}
	Now, for $(\rho,\ell,L)$-list recovery (with $\ell,L = O(1)$) one can crudely bound the relevant family $\mc F$ by $|\mc F| \leq 2^{O(n)}$; hence, by choosing $q$ to be large enough constant one can guarantee that the failure probability in the first bullet-point is indeed $q^{-\Omega(n)}$. For this reason, this approach is (as yet) only successful in the ``large alphabet regime.''
	
	Using this approach, the authors we able to derive exactly the threshold rate for $(\rho,L)$-list decoding (which we recall is the $\ell=1$ case of list recovery): it is the generalized Singleton bound, namely, $1-\rho\cdot(1+\frac1L)$.\footnote{There is an important distinction to be made between \emph{exactly} achieving the generalized Singleton bound (namely, having rate $1-\rho(1+\frac1L)$) and getting $\eps$-close (i.e., having rate $1-\rho(1+\frac1L)-\eps$). In particular, the field size must be exponentially large $2^{\Omega(n)}$ in the former case, but can be $2^{O(1/\eps)}$ in the latter. But we do not focus on these details here.} 
	
	\begin{remark}\label{rem:random-RS}
		In fact, even prior to~\cite{LMS24} it had been established that linear codes can achieve the generalized Singleton bound. Specifically, random Reed-Solomon codes had been shown to achieve this~\cite{BGM23,GZ23,AGL24}. This is a beautiful line of work, but a bit outside the scope of our survey, as these techniques have not (yet) been successfully applied to the case of list recovery. But, in light of the negative results we are about to present, that is perhaps unavoidable.
		
		Additionally, we would be remiss if we did not mention that a particularly technically involved portion of~\cite{LMS24} is devoted to proving that RLCs and random Reed-Solomon codes are \emph{locally equivalent}, which means that the threshold rate $R_{\mc P}$ for both of these random ensembles are the same. In particular, this means that if an RLC of rate $R$ is with high probability $(\rho,\ell,L)$-list recoverable, then the same is true for a RRS code.\footnote{There are some subtleties with the field size involved, but morally the above conclusion holds.} Thus, we now have two means of establishing that random Reed-Solomon codes achieve the generalized Singleton bound for list decoding. 
	\end{remark}
	
	\subsubsection*{The method of types (over small alphabet)}
	
	While LCL properties have been devised for analyzing random linear codes over large alphabets (and additionally, Reed-Solomon codes), it at least appears that such techniques will not be effective for small fields.
	
	As with LCL properties, we view sets of vectors of size $b$ as $b \times n$ matrices. Previously, we classified these matrices based on the linear dependencies they satisfy. Now, we take a more combinatorial, or information-theoretic, perspective: we classify them based on their empirical \emph{column} distribution. That is, let $\tau$ be the distribution over $\F_q^b$ obtained by sampling a uniformly random $i \in [n]$ and then outputting the $i$-th row of the matrix; a distribution obtained in this way is called a \emph{type}.\footnote{This notion is heavily inspired by the concept of \emph{typical sequences} as developed in information theory (and, in fact, can be viewed as a special case thereof). See, for instance, \cite[Chapter~3]{CT99}.} 
	
	Firstly, observe that the distributions obtained obtained in this way are in one-to-one correspondence with partitions of $[n]$ into $q^b$ sets, so there are at most 
	\[
	\binom{n+q^b-1}{q^b-1} \leq (n+1)^{q^b}
	\]
	such types. Note that if $q$ and $b$ are constants, then the above bound is actually \emph{polynomial} in $n$. But if $q$ is a growing parameter this bound quickly becomes unwieldy; hence the motivation for developing the LCL properties introduced above. 
	
	Now, as with LCL properties it is easy on can view local properties like list recovery as defined by forbidding certain families of types; see, for instance, \cite[Section~3.2]{resch2020}. Secondly, one can establish a \emph{threshold phenomenon} for all such types, just as with LCL properties. 
	
	Unfortunately, unlike in the case of LCL properties the method of types has only modest success in proving possibility results. In brief, while Mosheiff et al~\cite{MRRSW21} do characterize the threshold rate for all local properties, it is characterized as the value of an optimization, which unfortunately appears in general to be quite difficult to evaluate analytically (it does not even appear to be efficiently computable). Only for the following special cases has the threshold been computed exactly:
	\begin{itemize}
		\item $(\rho,2)$-list decoding with $q=2$~\cite{GMRSW21} and larger $q$~\cite{RY24} with $\rho < 1/3$. 
		\item $(\rho,3)$-list decoding with $q=2$~\cite{RY24}. 
	\end{itemize}
	However, as we will survey below, there has been some success using this framework for proving \emph{impossibility} results for list recovery of random linear codes. 
	
	\begin{remark}
		We would like to now clarify a potential point of confusion. Recall the capacity theorem \Cref{thm:LR-capacity-theorem}: it can be interpreted as answering the following question: ``what is the maximum rate of \emph{any} code that is $(\rho,\ell,L)$-list recoverable with $L$ \emph{not too large}?'' The threshold rates as defined above answer the following question: ``what is the maximum rate such that \emph{a random} linear code is likely $(\rho,\ell,L)$-list recoverable \emph{for a given $L$}?'' Thus, the capacity theorem gives coarser information, in the sense that it applies it only determines when list-recovery with small $L$ is possible, although it is stronger in that it applies to all codes, rather than a random ensemble of codes. 
	\end{remark}
	
	\begin{remark}
		While this is not the focus of our article, it is worth mentioning that the framework of types has been used fruitfully to show certain code ensembles are \emph{locally similar} to random linear codes. While we do not define this notion formally, it essentially boils down to demonstrating that the random code of interest contains constant sized subsets of vectors with roughly the same probability that a random linear code would. This is enough to show that the locally similar code has essentially the same threshold rate for all local properties (including list recovery). This was an original motivation of Mosheiff et al~\cite{MRRSW21}: they showed random low-density parity-check codes are locally similar to random linear codes. Followup works~\cite{GM22,PP24,MRSY24} have extended the list of code ensembles locally similar to random linear codes. 
	\end{remark}
	
	\begin{remark}
		For the case of PRCs of rate $R$, Guruswami et al~\cite{GMRSW21} also provide a toolkit for determining their threshold for local properties such as list recovery. In this case, it is much easier to evaluate the threshold rate (and in fact this work gives an efficient algorithm is given to compute it for list recovery). However it is quite likely that some codes do \emph{better} than plain random codes: indeed, this is already the case for minimum distance (i.e., $(\rho,1,1)$-list recovery), where random linear codes do better than plain random codes. (This is typically rectified by ``expurgating'' the plain random code, but we do not do that here.)
	\end{remark}
	
	\subsection{A long way with a few symbols}\label{sec:small-alphabet}
	
	We will now survey some results concerning the list recoverability of linear codes in the small alphabet regime. Here, the main focus will be \emph{random} linear codes (RLCs). We begin with \emph{the bad}, i.e., impossibility results. Afterwards we will share \emph{the good} (possibility results). Along the way, we will share some open problems that we consider to be interesting and approachable. 
	
	\subsubsection*{Impossibility results}
	
	The first result demonstrating an impossibility result for list recovery of linear codes is due to Guruswami et al~\cite{guruswami-li-moshieff-resch-2022-list-decoding-list-recovery-random-linear-codes}. Informally, they showed that if one is interested in list recovery \emph{from erasures} and the field has small characteristic, then random linear codes at capacity require exponentially large output lists. More precisely, suppose $q = \ell^t$ for some integer $t>1$, and that $\ell$ is a prime power (thus, $\F_q$ contains the finite field $\F_\ell$ as a subfield). Let $C$ be a random linear code for rate $R = 1-\rho - (1-\rho)\cdot \frac{1}{t}-\eps$; that is, $C$ is $\eps$-close to capacity (recall \eqref{eq:erasures-capacity} and note $\log_q\ell = 1/t$). Then (for small enough $\eps>0$) Guruswami et al~\cite{guruswami-li-moshieff-resch-2022-list-decoding-list-recovery-random-linear-codes} prove that with high probability a random linear code is \emph{not} $(\rho,\ell,\ell^{o(1/\eps)})$-list recoverable from erasures. This is in sharp contrast to \emph{plain} random codes, where lists of size $O(\ell/\eps)$ are sufficient. Thus, we see here a \emph{heavy price to pay for linearity!}
	
	What went wrong? For simplicity let us consider the case of $\rho=0$, i.e., zero-error list recovery. Intuitively, when $\F_q$ is an extension of $\F_\ell$, there are many ``somewhat linear'' input list tuples $(S_1,\dots,S_n)$ that the random linear code $C$ must avoid. For example, consider the input list tuple where each $S_i = \F_\ell$, or more generally some coset $\beta_i\cdot \F_\ell := \{\beta \cdot \alpha:\alpha \in \F_\ell\}$ for $\beta_i \neq 0$. Note that as $C$ is $\F_q$-linear by construction, it is also closed under $\F_\ell$ linear combinations, so if $x,y \in C \cap (\beta_1\F_\ell \times \cdots \times \beta_n \F_\ell)$, then additionally $\gamma x + \delta y \in C \cap (\beta_1\F_\ell \times \cdots \times \beta_n \F_\ell)$ for all $\gamma,\delta \in \F_\ell$. Extending this line of reasoning, once $C \cap (\beta_1\F_\ell \times \cdots \times \beta_n \F_\ell)$ contains $\Omega(1/\eps)$ $\F_\ell$-linearly independent vectors, it will follow that $|C \cap (\beta_1\F_\ell \times \cdots \times \beta_n \F_\ell)| \geq \ell^{\Omega(1/\eps)}$. Exploiting the method of types, Guruswami et al~\cite{guruswami-li-moshieff-resch-2022-list-decoding-list-recovery-random-linear-codes} successfully establish that this is indeed likely to happen, establishing the impossibility result for list recovery. 
    
    Before continuing, we mention the following open problem. 
    \begin{question}
        Does the same lower bound $L \geq \ell^{\Omega(1/\eps)}$ apply to \emph{every} linear code?
    \end{question}
    As we see in \Cref{thm:LR-outlist-lower-bound}, in the large field case such a lower bound indeed applies to \emph{every} linear code; in fact, it applies for list recovery from erasures \emph{and} corruptions, and moreover over \emph{all} (large enough) fields.\footnote{Perhaps unsurprisingly, the proof of \Cref{thm:LR-outlist-lower-bound} (which we reproduce) shares many similarities with the argument sketched above.} Here, we emphasize that the lower bound of~\cite{guruswami-li-moshieff-resch-2022-list-decoding-list-recovery-random-linear-codes} \emph{only applies} to list recovery from erasures and only if the field size is an integer power of the input list size; if, say, $q$ is prime then this lower bound does not apply. 
	
	Next, we can consider what happens for list recovery \emph{from corruptions}. In this case, Resch and Yuan~\cite{RY24} show that for RLCs of rate $\Rcor-\eps$ the list size cannot be taken smaller than 
	\begin{align}
		\frac{\log_q\binom q\ell -(1-h_{q,\ell}(\rho))}{\eps} \ . \label{eq:list-rec-linear-lb}
	\end{align}
	This argument again uses the method of types, and in fact generalizes an argument given by Guruswami et al~\cite{guruswami-li-moshieff-resch-2022-list-decoding-list-recovery-random-linear-codes} that only applied to list decoding. We remark that this is of order $O(\ell/\eps)$; that is, this lower bound does not demonstrate any price to pay for linearity. 
	
	In fact, Resch and Yuan~\cite{RY24} conjecture that this bound is \emph{tight} for RLCs. Furthermore, they show that for plain random codes (PRCs) the list size must be (with high probability) essentially
	$\frac{\log_q\binom q\ell}{\eps}$.
	That is, if this conjecture is correct, RLCs will actually have \emph{better} list recovery than PRCs!
	
	While this might appear surprising at first, this is in fact in line with what happens with the Gilbert-Varshamov (GV) bound. Recall that the GV bound states that there exist codes with rate $R = 1-h_q(\delta)-\eps$ that have minimum distance $\delta$. One way to prove this is to consider a RLC of rate $R$ and argue that, with high probability, it does not contain a nonzero vector of weight $\leq \delta n$. If one considers a PRC, it is not sufficient to just consider low weight codewords: one must consider pairs of nearby vectors, and argue none of these pairs are contained in the code. This requires a union bound over significantly more ``bad events,'' and because of this one cannot choose the rate as large.\footnote{This can be fixed by ``expurgating'' the code: namely, throwing out a small number of codewords to fix the minimum distance without significantly harming the rate. But this expurgated code is no longer a ``plain'' random code.} Thus, for the simple local property of ``minimum distance $\delta$'' -- i.e., $(\delta/2,1,1)$-list recoverability -- RLCs perform better than PRCs. 
	
	Additionally, for the case of list decoding over the binary alphabet, Li and Wootters~\cite{LW20} establish that RLCs perform better than PRCs. The intuition is that for linear codes, centers that differ by a codeword behave the same in terms of their list size: if $z-z'\in C$, then $|B(z,\rho) \cap C| = |B(z',\rho) \cap C|$. This cuts down on the size of the union bound roughly by a factor $q^{Rn}$, which allows one to then choose the list size $L$ smaller (roughly, one can subtract $R/\eps$ from $L$). Thus, the available evidence is that, over small alphabets, RLCs might be best for list recovery. 
	
	On this positive note, let's now turn to some possibility results. 
	
	\subsubsection*{Possibility results}
	
	Until recently, except for the Zyablov-Pinsker argument we did not have \emph{any} positive results for the small alphabet regime, at least if one insists on codes $\eps$-close to capacity. This has now changed as Doron et al~\cite{DMRR25} -- following the approach H{\aa}stad, Guruswami and Kopparty~\cite{GHK10} devised for list-decoding -- have established that RLCs of rate $1-h_{q,\ell}(\rho)-\eps$ are whp $(\rho,\ell,C_{\rho,\ell,q}/\eps)$-list recoverable (from corruptions). Thus, if one assumes $q$ (and hence $\ell$) to be constant (and furthermore that $\rho$ is bounded away from $1-\ell/q$) then in fact the output list size is not significantly worse than what we know to be possible existentially, at least in terms of the gap-to-capacity $\eps$. In a nutshell: in the small alphabet regime, up to constants there is \emph{no price for linearity}!
	
	Now, it would be nice to have $C_{q,\ell,\rho}$ close to the lower bound of \eqref{eq:list-rec-linear-lb}; in particular, at most $O(\ell)$. Unfortunately the argument of Doron et al~\cite{DMRR25} is unable to establish this: they obtain roughly $C_{q,\ell,\rho} = q^{O(\ell \log^{O(1)} q)}$ (assuming $\rho$ is not too close to $0$ or $1-\ell/q$). Still, assuming $q \leq 2^{(1/\eps)^c}$ for some small universal constant $c>0$, this improves upon the list-size guaranteed by Zyablov-Pinsker (which we recall was $q^{\ell/\eps}$). 
	
	Next, Doron et al~\cite{DMRR25} consider list recovery of erasures. In light of the result of Guruswami et al~\cite{guruswami-li-moshieff-resch-2022-list-decoding-list-recovery-random-linear-codes}, one might suspect that there is nothing positive one can say: $L$ must be exponential in $1/\eps$ for RLCs $\eps$-close to capacity. However, recall the argument of Guruswami et al~\cite{guruswami-li-moshieff-resch-2022-list-decoding-list-recovery-random-linear-codes} heavily used the assumption that $q$ was an integer power of $\ell$, which itself was assumed to be a prime power. However, if $q$ is, say, prime, this argument completely breaks down. Doron et al~\cite{DMRR25} show this is \emph{inherent}: an RLC $\eps$-close to capacity over a \emph{prime} field is with high probability $(\rho,\ell,L)$-list recoverable from erasures with $L = C_{\rho,\ell,q}/\eps$. The bound on $C_{\rho,\ell,q}$ is roughly $q^{O(\ell \log q})$ (assuming $\ell\leq0.99 q$, say), which is quite far from $O(\ell)$ -- that is, this argument does \emph{not} show RLC's achieve the Elias bound -- but it does again show that $L$ can just have linear dependence on $1/\eps$. 
	
	To briefly explain the argument, recall again the Zyablov-Pinsker argument: namely, for every subset $X$ of size $L+1$ of $B_\rho(S)$ (for list recovery from corruptions) or a combinatorial rectangle $S_1 \times \cdots \times S_n$ (for erasures), they conclude that $\Pr[X \subseteq C] \leq q^{-(1-R)\log_q(L+1)}$. This bound is tight if the subset $X$ has a lot of ``linear structure,'' i.e., if it looks a lot like a subspace. Doron et al~\cite{DMRR25} show that $B_\rho(S)$ and $S_1 \times \cdots \times S_n$ (the latter only when $q$ is prime) do not have much linear structure. More precisely, they argue that if one takes two independent random samples from these sets, any linear combination of these samples is unlikely to again lie in a fixed shift of these sets. Once this is established, they could follow ideas of Guruswami, H{\aa}stad and Kopparty~\cite{GHK10} to obtain their list recovery results. 

    Finally, recalling the lower bound from \eqref{eq:list-rec-linear-lb}, the following question is natural.

    \begin{question} \label{ques:list-size-small-alpha}
        Is the lower bound \eqref{eq:list-rec-linear-lb} tight? More modestly, can one establish that $L=O(\ell/\eps)$ is sufficient for RLCs (namely, that they achieve the Elias bound)? Even more modestly, any bound that has just \emph{polynomial} dependence on $\ell$?
    \end{question}

	\subsection{A long(er) way with many symbols}\label{sec:large-alphabet}
	
	We will now survey some more recent results in the setting of large fields.  Historically, an advantage we have had in this setting is a lot more explicit constructions that continue to get better.  A representative class of such constructions are the popular families of polynomial codes.  We do not cover results on these in detail, and so will only define them informally.

    Fix a finite field \(\F_q\), and take distinct points \(\alpha_1,\ldots,\alpha_n\in\F_q\).  For any \(k\in[n]\), the \tsf{Reed-Solomon (RS) code} consists of message space \(\F_q^k\) interpreted as the space of coefficient vectors \((f_0,\ldots,f_{k-1})\) of polynomials \(f(X)=f_0+f_1X+\cdots+f_{k-1}X^{k-1}\), and the corresponding codewords are evaluation vectors \((f(\alpha_1),\ldots,f(\alpha_n))\).  It is elementary to see that these codes are MDS, making them prime candidates for investigating list decodability and list recoverability.  As it turns out, we know that all RS codes can be algorithmically list recovered up to the \emph{Johnson bound}~\cite{guruswami-sudan-1998-RS-list-decoding} with constant output list size, but this radius is significantly smaller than capacity.  We know more (combinatorially) when the evaluation points are chosen randomly (cf.~\cref{rem:random-RS}), but for algorithmic progress, the better families of codes (as of now) are the \emph{folded variants} -- Folded Reed-Solomon (FRS) codes and multiplicity codes.

    For the folded variants, there is an additional folding parameter \(s\ge2\).  The \tsf{FRS code} again consists of a similar message space, but the codewords are formed by larger symbols \((f(\alpha_i),f(\gamma\alpha_i),\ldots,f(\gamma^{s-1}\alpha_i)),\,i\in[n]\), where \(\gamma\in\F_q^\times\) is a fixed element with high multiplicative order.  The \tsf{multiplicity code} also consists of a similar message space, but the codewords are formed by larger symbols \((f(\alpha_i),f'(\alpha_i),\ldots,f^{(s-1)}(\alpha_i)),\,i\in[n]\).  The intuition obtained by the better algorithms~\cite{guruswami-rudra-2008-FRS,kopparty-2015-multiplicity-code,guruswami-wang-2013-FRS,kopparty-ron-zewi-saraf-wootters-2023-list-decoding} is that \emph{sufficiently large folding enables better algorithmic decoding}.
	
    To get a sense of what is possible, it has been established that FRS codes of rate $1-R-\eps$ are $\big(\rho,\ell,(\ell/\eps)^{O(\frac{\log\ell}{\eps})}\big)$-list recoverable~\cite{kopparty-ron-zewi-saraf-wootters-2023-list-decoding,tamo-2023-tighter-list-size} (and furthermore such a code comes equipped with an efficient list recovery algorithm). Could we hope to do better -- \emph{lower the upper bound}? 

    Complementary to upper bounds, a key turning point in the state-of-the-art was the recent work of~\cite{chen-zhang-2024-FRS-list-size}, which gives a lower bound \(\ell^{\,\Omega(1/\epsilon)}\) on the output list size for the Reed-Solomon code and its folded variants.  Once again, could we hope to do better -- \emph{raise the lower bound}?
	% \paragraph{Linear codes.}  
	% Begin with the easy lower bound of~\cite{LS25} for all linear codes over large alphabet, and quickly show the proof.
	
	% Mention the upper bound, and give a proof hint.  More details can be mentioned in the special case of polynomial codes.
	
	% Show the \emph{potentially weak} lower bound that the same proof gives over small alphabet, as in Nic's paper.  Raise the conjectured stronger lower bound therein.  Mention the known upper bound therein (without proof, I guess).
	
	\subsubsection*{Lower bounds}
	
	%Among the most popular codes in the large alphabet regime, with amazing amenability for algorithmic decoding, are the families of polynomial codes. 
    Interestingly, for the list recovery of polynomial codes, we barely knew of any lower bounds for a very long time.  The situation has changed drastically changed in the last couple of years, and we focus on this.  We will show a recent lower bound on output list size for list recovery of \(\mb{F}_q\)-linear codes, which is by far the most important recent development of late.  Complementary to what we discussed in~\cref{sec:small-alphabet}, this result shows that there is a \emph{price to pay for linearity} over large fields too!  This argument is due to~\cite{LS25} for list recovery up to capacity, but we make some minor modifications and present it for the case of zero-error list recovery.  Note that a lower bound for zero-error list recovery is also a lower bound for list recovery up to any other radius; so technically, this is a stronger result.
	\begin{theorem}[{\cite{LS25}}, adapted to zero-error list recovery]\label{thm:LR-outlist-lower-bound}
		Consider the finite field \(\mb{F}_q\) with \(q=\ell^{\,t}\).\footnote{To clarify, we do not assume \(\ell\) is a prime or a prime power.  In full generality, we have a prime \(p\), and a prime power \(q=p^r,\,r\ge1\), which can then be written as \(q=p^r=\ell^t\), where \(t>1\) is a real number and \(\ell\ge2\) is an integer.}  For any \(\epsilon\in\big[0,\frac{1}{2}-\frac{1}{t}\big]\), and \(n\ge n_0(t,\epsilon)\) sufficiently large, if \(C\subseteq\mb{F}_q^n\) is an \(\mb{F}_q\)-linear code of rate \(1-\frac{1}{t}-\epsilon\) that is \((0,\ell,L)\)-list recoverable, then \(L\ge\ell^{\,\frac{1}{2}\min\{t,1/\epsilon\}}\).
	\end{theorem}
	
        It is worth noting that the claim of~\cref{thm:LR-outlist-lower-bound} is a more abstracted version motivated by the proof in the polynomial setting by~\cite{chen-zhang-2024-FRS-list-size}.  And of course, an answer to a question leads to another question!
    \begin{question}\label{que:LB-small-alphabet}
		Can the assertion in~\cref{thm:LR-outlist-lower-bound} be improved to \(L\ge\ell^{\,\Omega(\max\{t,1/\epsilon\})}\)?
	\end{question}
	
	\begin{proof}[Proof of{~\cref{thm:LR-outlist-lower-bound}}]
		Note that we assume \(n\ge n_0(t,\epsilon)\) is sufficiently large, and we will determine \(n_0\) at the end.  Let \(k=\dim_{\mb{F}_q}(C)=\big(1-\frac{1}{t}-\epsilon\big)n\).  By using elementary column operations, and permutations of rows, we can write a generator matrix \(G\in\mb{F}_q^{n\times k}\) of \(C\) in \emph{reduced column-echelon form} as
		\begin{align*}
			G&=\begin{bmatrix}
				G_1&G_2&\cdots&G_{k-1}&G_k
			\end{bmatrix}\\
			&\coloneqq\begin{bmatrix}
				1&&&\\
				&1&&\\
				&&\ddots&&\\
				&&&1&\\
				&&&&1\\
				\hline
				g_1&g_2&\cdots&g_{k-1}&g_k
			\end{bmatrix}
		\end{align*}
		for some \(g_1,\ldots,g_k\in\mb{F}_q^{(n-k)\times1}\).  Fix \(m\ge1\) (to be determined), and for every \(j\in[m]\), denote
		\begin{align*}
			A_j&=\begin{bmatrix}
				g_{(j-1)m+1}&\cdots&g_{jm}
			\end{bmatrix}\in\mb{F}_q^{(n-k)\times(k/m)},\\
			B_j&=\begin{bmatrix}
				G_{(j-1)m+1}&\cdots&G_{jm}
			\end{bmatrix}\in\mb{F}_q^{n\times(k/m)}.
		\end{align*}
		
		Suppose \(n-k<k/m\).  For every \(j\in[m]\), there exists a nonzero vector \(x^{(j)}\in\mb{F}_q^{(k/m)\times 1}\) such that \(A_jx^{(j)}=0^{n-k}\).  Define \(h_j\coloneqq G_jx^{(j)}\in\mb{F}_q^n\) for all \(j\in[m]\).  It follows that \(h_j\in C,\,h_j\ne0^n\), and \(\supp(h_j)\subseteq[(j-1)(k/m)+1,j(k/m)]\) for all \(j\in[m]\).
		
		Now take any distinct \(\beta_1,\ldots,\beta_\ell\in\mb{F}_q\), and consider the grid \(B^m\coloneqq\{\beta_1,\ldots,\beta_\ell\}^m\).  Define
		\[
		\mc{L}=\left\{\sum_{j=1}^m\gamma_jh_j:(\gamma_1,\ldots,\gamma_m)\in B^m\right\}\subseteq C.
		\]
		It is immediate that \(|\mc{L}|=\ell^{\,m}\).  Let \(S_i=\{\gamma_rh_j(i):\gamma_r\in B\}\) for all \(j\in[m]\) and \(i\in[(j-1)(k/m)+1,j(k/m)]\).  Also choose arbitrary \(S_i\in\binom{\mb{F}_q}{\ell}\) such that \(0\in S_i\), for all \(i\in[k+1,n]\).  So we have \(\mc{L}=S_1\times\cdots\times S_k\times\{0\}^{n-k}\subseteq C\cap(S_1\times\cdots\times S_n)\).
		
		All that remains now is to optimize \(m\ge1\) subject to the condition \(n-k<k/m\).  We have
		\begin{align*}
			&\n n-k<\frac{k}{m}\\
			\iff&\n\bigg(\frac{1}{t}+\epsilon\bigg)n<\frac{\big(1-\frac{1}{t}-\epsilon\big)n}{m}\\
			\iff&\n m<\frac{1-\frac{1}{t}-\epsilon}{\frac{1}{t}+\epsilon}.
		\end{align*}
		So we choose
		\[
		m=\bigg\lceil\frac{1}{\frac{1}{t}+\epsilon}\bigg\rceil\begin{cases}
			\ge\frac{t}{2}&\tx{if }\epsilon>0,\,t\le\frac{1}{\epsilon},\\[5pt]
			\ge\frac{1}{2\epsilon}&\tx{if }\epsilon>0,\,t>\frac{1}{\epsilon},\\[5pt]
			=t&\tx{if }\epsilon=0.
		\end{cases}
		\]
		This implies \(|\mc{L}|\ge\ell^{\,\frac{1}{2}\min\{t,1/\epsilon\}}\).
		
		To conclude, let us determine \(n_0\).  Assume the minimal case \(k_0=\big(1-\frac{1}{t}-\epsilon\big)n_0\).  Since we need \(k_0/m\ge1\), it is enough to ensure
		\[
		n_0\ge\frac{\frac{1}{\frac{1}{t}+\epsilon}}{1-\frac{1}{t}-\epsilon}=\frac{1}{\big(\frac{1}{t}+\epsilon\big)\big(1-\frac{1}{t}-\epsilon\big)}.
		\]
		Since we also have
		\[
		\frac{1}{\big(\frac{1}{t}+\epsilon\big)\big(1-\frac{1}{t}-\epsilon\big)}\le2\max\bigg\{t,\frac{1}{\epsilon}\bigg\},
		\]
		we choose \(n_0=n_0(t,\epsilon)=2\max\{t,\lceil1/\epsilon\rceil\}\), and this completes our proof.
	\end{proof}

    Now, the astute reader may have already noted that FRS codes are \emph{not technically linear}: the alphabet is in fact a vector space over the base field $\F_q$. They are, however, \emph{additive} over the base field $\F_q$. This motivates the following question of whether we can reconcile the linear and additive settings.

    \begin{question}
        Can the lower bound of Li and Shagrithaya~\cite{LS25} be generalized to \emph{additive} codes? 
    \end{question}
    
	\subsubsection*{Upper bounds}
	
    The best progress on upper bounds for output list size has also come from the recent work of~\cite{LS25}.  Previously, it was known via algorithmic list recovery that FRS codes and multiplicity codes having rate \(R\) can be list recovered up to capacity \(1-R-\epsilon\) (for sufficiently large folding), with output list size \((\ell/\epsilon)^{O(\log\ell/\epsilon)}\).  The key observation behind these results was that the output list size can always be captured inside a small-dimensional (of course, depending on \(\ell\)) subspace of the code.  The punchline of~\cite{LS25}, who adapted the Zyablov-Pinsker argument (outlined in~\cref{sec:ZP}), is that even for random linear codes, the output list for list recovery up to capacity is contained in a small-dimensional subspace, and therefore the output list size is \((\ell/\epsilon)^{O(\ell/\epsilon)}\).

	\subsubsection*{Random linear codes} 
	
	Let us now return to determining what is possible existentially for random linear codes. Namely, what is the ``typical'' list recoverability of a linear code? Recall that in light of \cite{LS25} we know that if we hope to list recover up to radius $1-R-\eps$ then the output list size must be at least $\ell^{\Omega(1/\eps)}$ (assuming $\eps$ is small enough). However, if one backs off \emph{slightly} more from capacity, then improved list sizes could be possible. 
	
	The first work making progress on this question is due to Rudra and Wootters~\cite{rudra-wootters-2018-list-recovery}. They provide a general framework for understanding list recovery of RLCs and, among other results, show the following.
	\begin{itemize}
		\item When $\rho$ is very close to $1-\ell/q$ -- viz, $\rho=1-\ell/q-\delta$ for sufficiently small $\delta>0$ -- and the rate is of the form $(1-\eps)\cdot \Rcor$, the list-size $L$ may be bounded by $L \leq q^{O(\log^2(\ell/\delta))}$. In brief: in the high-noise regime, with a \emph{multiplicative} gap to capacity, list sizes of size \emph{quasipolynomial} in $\ell$ suffice. 
		\item Consider now the ``high-rate, low noise'' regime. That is, the rate $R=1-\gamma$ (for sufficiently small $\gamma>0$). Then, for $\rho = \Omega(\gamma)$ one can show RLC's are with high probability $(\rho,\ell,L)$-list recoverable with $L = (q/\gamma)^{O(\log^2\ell/\gamma^3)}$. That is, one can have high rate list recoverable codes with near-optimal decoding radius $\rho$ and again quasipolynomial list size.  Earlier Guruswami~\cite{guruswami-2004-thesis} had shown that in this high-rate regime the list size $L$ may be bounded by $\ell^{O(\ell/\gamma^2)}$.
	\end{itemize}

    The argument of Rudra and Wootters proceeds by  considering bad sets of vectors (i.e., large subsets of list recovery balls) and argues that there are not so many of them that satisfy many linear constraints. They provide a recursive argument: namely, if there is some ``bad, low-dimensional'' message set, then they can (with high probability) find another smaller ``bad, low-dimensional'' message set. Iterating this argument enough times, they eventually find that the code contains whp a very small bad, low-dimensional set, which cannot exist (e.g., dimension less than logarithmic in the set size). Crucial to their argument is a useful proxy for dimension which is somewhat reminiscent of a moment generating function which allows for easier analytic control. 

    \paragraph{Rumble in the jungle: linearity v/s additivity.}  At this point, the astute reader has probably noticed that the list sizes guaranteed by Rudra and Wootters~\cite{rudra-wootters-2018-list-recovery} do not improve upon the list-size guaranteed for the \emph{explicit} FRS code! So, why bring up these results at all? Well, once again because FRS codes are additive, and not technically linear.

    Now, in many regards the property of additivity is just as useful as linearity; for example, either way one can efficiently encode using a generator matrix. However, for certain operations involving codes, such as tensoring, or in many ``expander-based'' constructions, such as the \cite{alon-edmonds-luby-1995-AEL} construction shared earlier, additivity appears somewhat naturally. At the same time, in such constructions the final output list size of the constructed code is inherited from the ``building block'' codes; thus, one should hope for the inner codes to have the smallest list size possible! Additionally for constructing \emph{quantum} error-correcting codes, linearity is a typically basic requirement. Thus, linearity and additivity are intimately tied together in several construction and proof strategies, and therefore, properly understanding what linear codes can and cannot achieve in the context of list recovery is vital. 
    
	This largely completes the story for list recovery of random linear codes. There are naturally many interesting questions remaining to be studied; we list a couple now.
	
	\begin{question}\label{que:L-poly-in-ell}
		\begin{enumerate}[(1)]
			\item For small $\gamma>0$, do there exist codes of rate $1-\gamma$ that are $(\Omega(\gamma),\ell,L)$ list recoverable with $L=\poly(\ell)$? Either from corruptions or errors? Or even for zero-error? 
			\item For a broader range of parameters, do there exist linear codes of rate $0.99\Rcor$ (say) that are $(\rho,\ell,L)$-list recoverable with $L$ \emph{polynomial} in $\ell$? Similarly for erasures, or even the zero-error case? 
		\end{enumerate}
	\end{question}
	
	%Introduce LCL properties, and prove that list decoding and list recovery can be captured here.
	
	%Proof outline of how list recovery up to capacity for random linear codes implies the same for random RS codes. \venki{How much details should I give here?  I think the main takeaway of this section is just that list recovery can be captured by the LCL framework.}

	\section{List recovery beyond coding theory} \label{sec:beyond-coding}

    Now that we have sharpened our blade, let us see how we could use it.  With the bulk of our intended discussion done, we now move towards a brief look at the applicability of list recovery in a couple of ancillary areas.
	
	\subsection{Leakage resilience of secret-sharing schemes}
	
	Linear codes play important roles in many areas of cryptography. While there are many connections we could discuss, we will focus on (local) leakage-resilience of secret-sharing schemes, as it demonstrates an interesting connection to list recovery, and points towards challenges that current techniques cannot address. 
	
	Consider the following scenario: for some $1 \leq t< n$, a secret $s \in \F_q$ is chosen, and then we wish to ``share'' it among $n$ parties such that:
	\begin{itemize}
		\item Reconstruction: any collection of $(t+1)$ (or more) parties can reconstruct their secret by pooling their shares;
		\item Privacy: any collection of $t$ (or less) parties learn nothing about the secret.
	\end{itemize}
	
	The classic way (called \emph{Massey secret sharing}~\cite{M95}) to do this is to start with $C' \leq \F_q^{n+1}$ which is \emph{maximum distance separable (MDS)} of dimension $t+1$. To share a secret $s$, sample $s_1,\dots,s_n$ uniformly at random subject to $(s,s_1,\dots,s_n) \in C'$, and then give party $i$ the share $s_i$. Briefly, reconstruction uses the fact that every subset of $[n+1]$ size at $t+1$ is an information set for $C$, while privacy uses the fact that the dual of an MDS code is also MDS. The special case of this construction with $C'$ being a Reed-Solomon code yields the celebrated Shamir secret-sharing scheme~\cite{S79}. 
	
	Now, the above setting implicitly assumes that we have an adversary that can control a \emph{strict} subset of the parties, and that the adversary obtains their shares exactly. However, there are many practical attacks that exploit hardware \emph{side-channels}, such as reading power consumption, cache-access patterns, time used, etc. Such attacks can be performed against even trusted parties; however, they are typically unable to fully recover a party's share, but rather some information about it (say, a few bits). 
	
	In light of these threats, cryptographers have begun to define abstract models of \emph{leakage-resilient cryptography}~\cite{KR19}, with a particularly well-studied topic being \emph{leakage-resilient secret-sharing schemes}. As a simple case (called \emph{local} leakage), we consider the following model: the adversary chooses leakage functions $g_i:\F_q\to\{0,1\}$ for each share $i\in[n]$, and learns the bit $g_i(s_i)$ (where party $i$ receives share $s_i$). 
	
	We would like to argue that the adversary is unable to learn anything about the share $s$ from these leakages. One way to argue this is to prove that the distribution of $(b_1,\dots,b_n)$ is essentially uniform over $\{0,1\}^n$. 
	
	Now, what does all this have to do with list recovery? Note that for $(g_1(s_1),\dots,g_n(s_n))$ to be close to uniform, we need that for each $(b_1,\dots,b_n) \in \{0,1\}^n$, $\Pr[(g_1(s_1),\dots,g_n(s_n))=(b_1,\dots,b_n)] \approx 2^{-n}$. Observe that, if we call $S_i := g_i^{-1}(b_i) \subseteq \F_q$, we have
	\begin{align*}
		&\Pr[(g_1(s_1),\dots,g_n(s_n))=(b_1,\dots,b_n)] \\
		&\qquad\qquad= \frac{|C \cap (S_1\times \cdots \times S_n)|}{q^k}
	\end{align*}
	where we have defined the code $C:=\{(s_1,\dots,s_n):(0,s_1,\dots,s_n) \in C'\}$.\footnote{It is a standard argument that it suffices to consider shares of the secret $0$.} Note that the numerator is precisely the quantity one must bound to demonstrate $C$ is zero-error list-recoverable! While \emph{a priori} our task seems stronger in the sense that we now need a two-sided bound, often at the cost of an affordable degradation in parameters just proving an upper bound suffices. The main conceptual change is that we are now explicitly hoping to list recover \emph{above capacity}; the output list-size $L$ will be $q^{\Theta(n)}$, and the goal is to have $L$ always be roughly what one would expect, namely, $\frac{q^k\ell^n}{q^n}$.
	
	Recall now that we have certain impossibility results for linear codes over zero-error list recovery over small characteristic fields. And indeed, Guruswami and Wootters~\cite{GW16}\footnote{They're motivation was in fact the construction of a so-called local repair scheme, but along the way they essentially gave a successful local leakage attack on any linear secret sharing scheme over characteristic 2 fields.} showed that over characteristic 2 fields, leaking a single bit from each of the shares is enough to recover a bit of the secret. That is, such leakage-resilient secret sharing schemes cannot exist. 
	
	Thus, the question that recent works~\cite{BDIR18,MNPS21,MNPW22,KK23} have addressed is whether or not one can have leakage-resilient schemes over prime fields. The state-of-the-art shows that RLCs of rate $R>1/2$ are leakage-resilient (and indeed, at its core is a certain argument establishing list recovery of RLCs, albeit in a nonstandard regime). For random RS codes, the best achievable rate is currently $\approx 0.69$. (Some other works manage to rule out large families of leakages for codes of smaller rates, but none rule out \emph{all} attacks for codes of smaller rate.) It is known the rate $R$ must satisfy $\frac{R}{1-R} > \frac{2}{\log_2 q}$~\cite{NS20}; the general consensus is that leakage-resilience should at least be possible for any positive rate $R>0$ for large enough $q$ (and perhaps the lower bound is even tight). 

    \begin{question} \label{ques:secret-sharing}
        Is there a linear secret sharing scheme of rate $R \leq 1/2$ that is leakage-resilient? Perhaps this holds for any $R>0$ (assuming $q$ large enough)?

        For RS codes, can you get rate smaller than $0.69$?
    \end{question}
	
	It is worth discussing where the challenge arises in proving such a leakage resilience result. Consider the case where each leakage function $g_i$ is roughly balanced, so each $S_i = g_i^{-1}(b_i)$ is roughly of size $q/2$. That is, the input list size $\ell \approx q/2$. Recall in the proof of \Cref{thm:LR-capacity-theorem} that we took a union bound over all possible tuples $(S_1,\dots,S_n) \in \binom{\F_q}{\ell}^n$, which we bounded crudely by $\ell^n$. When $\ell\approx q/2$, the best bound we can give is now $2^{qn}$; since must think of $q>n$ here (recall we would like $C$ to be MDS), this is typically prohibitively large. Thus, prior works have needed to find ingenious arguments to outperform this simple union bound.

    \subsection{Pseudorandomness}

        To cap things off, we return to the glimpse of pseudorandomness that we had in the warmup~\cref{sec:AEL}.  We had seen how expander graphs can be used as lifting tools to improve the rate v/s distance tradeoff of concatenated codes.  \emph{Lo and behold, expander graphs are equivalent to list recoverable codes!}

        To understand this beautiful connection, we need to consider a version of list recovery that is easier to handle than classic list recovery -- we replace combinatorial rectangles with restricted subsets.  Let \(C\subseteq\Sigma^n\) be a code.  Given radius \(\rho\in(0,1)\), and integers \(\ell,L\ge1\), we say \(C\) is \tsf{\((\rho,\ell,L)_{\ell^1}\)-list recoverable} if we have the intersection size \(|B(\rho,S)\cap C|\le L\), for all combinatorial rectangles \(S=S_1\times\cdots\times S_n\) satisfying \(|S_1|+\cdots+|S_n|\le\ell\).\footnote{This version of list recovery does not have a standard name in the literature; in fact, it is usually just called list recovery, with the distinction being obvious from the context.  Since the combinatorial rectangles are restricted by an \(\ell^1\)-restriction on their size, rather than an \(\ell^\infty\)-restriction, we call this \(\ell^1\)-list recovery in this survey.}  Notice that we may have some \(S_i=\emptyset\), in which case we simply replace it to \(S_i=\{\perp\}\) and consider an erasure at that location.  (So, strictly, speaking, we are allowing list recovery from erasures.)

        \paragraph*{And they're back: folded codes \`a la additive codes.}  Once again, we will nee to consider additive codes.  We present here another interpretation in terms of folding, since it is equally widespread -- in fact, the standard polynomial code families are defined in terms of folding rather than additivity.  Folded codes allow for bunching of codeword symbols, and the distances are measured in a coarser sense.  Formally, a folded code is given by an encoding map \(C:\Sigma^k\to(\Sigma^s)^n\), where \(s\) is the \emph{folding parameter}, the Hamming distance is now measured with respect to the alphabet \(\Sigma^s\), the message length \(k\le sn\), and the length of the code is still \(n\).

        This following connection between expander graphs and list recoveable codes was noted in~\cite{guruswami-umans-vadhan-2009-unbalanced-expanders}, and their work contains a lot more; we only provide a motivating glance.  In a slight departure from the definition in~\cref{sec:AEL}, we will now consider bipartite graphs with \emph{vertex expansion}.  Let \(G=(U,V,E)\) be a bipartite graph.  (We don't assume \(|U|=|V|\).)  For \(K\ge1,\,\mu>0\), we say \(G\) is a \tsf{\((K,\mu)\)-vertex expander} if for every subset \(A\subseteq U\) with \(|A|\le K\), we have \(|\tx{\normalfont Nbr}(A)|\ge\mu|A|\).
        \begin{theorem}[{\cite[Lemma 3.1]{guruswami-umans-vadhan-2009-unbalanced-expanders}}]\label{thm:expander-iff-LR}
            Let \(C:\Sigma^k\to(\Sigma^s)^n\) be a folded code.  Define a bipartite graph \(G_C=(\Sigma^k\times[n],[n]\times\Sigma^s,E)\) by
            \[
            \big\{(a,i),(j,b)\big\}\in E\n\iff\n j=i,\,C(a)_i=b.
            \]
            Then \(C\) is \((0,\ell,L)_{\ell^1}\)-list recoverable if and only if \(G_C\) is a \((L,\ell/L)\)-vertex expander.
        \end{theorem}
        \begin{proof}
            Let us note the contrapositive condition to vertex expansion.  \(G_C\) is a \((L,\ell/L)\)-vertex expander if and only if for any \(B\subseteq[n]\times\Sigma^s\) and \(A\subseteq\Sigma^k\times[n]\) such that \(\tx{\normalfont Nbr}(A)\subseteq B\), if \(|B|<\ell\) then \(|A|<L\).  Does this sound like list recovery?  (Yes!) Because it is.  The correspondence follows immediately.
        \end{proof}

        While~\cref{thm:expander-iff-LR} is elementary, it is a classic instance to note the stark departure in the parameters of list recovery from the coding theory setting.  Two notable departures are the following.
        \begin{itemize}
        \item  In coding theory, the length of the list recoverable code is growing, and the alphabet size is often smaller (at least when we consider the small alphabet regime).  When we consider vertex expanders, we would like to keep the length of the code as small as possible (hopefully, a constant), and an infinite family of expander graphs is obtained precisely due to a growing alphabet.

        \item  In coding theory, while we are definitely interested in optimal output list size relative to the input list size, even a polynomial loss in an explicit construction is not severe.  When we consider expander graphs, such a loss could be devastating since the expansion factor \(\ell/L\) could become negligible.
        \end{itemize}

        In fact,~\cref{thm:expander-iff-LR} is only a beginning. There are several connections between pseudorandomness and list recovery; to pique the reader's curiosity, we point to a few more -- different versions of list recovery have an essentially one-to-one correspondence between \emph{condensers} and \emph{strong extractors}.  We point the reader to~\cite{guruswami-umans-vadhan-2009-unbalanced-expanders} for the connection with condensers, and~\cite{TUZ01} for the connection with strong extractors.  While we do not delve on these objects here, and we do not even define them, see~\cref{fig:LR-bijections} for a hint of the correspondences.
	\begin{figure*}[h]
		\centering
		\includegraphics{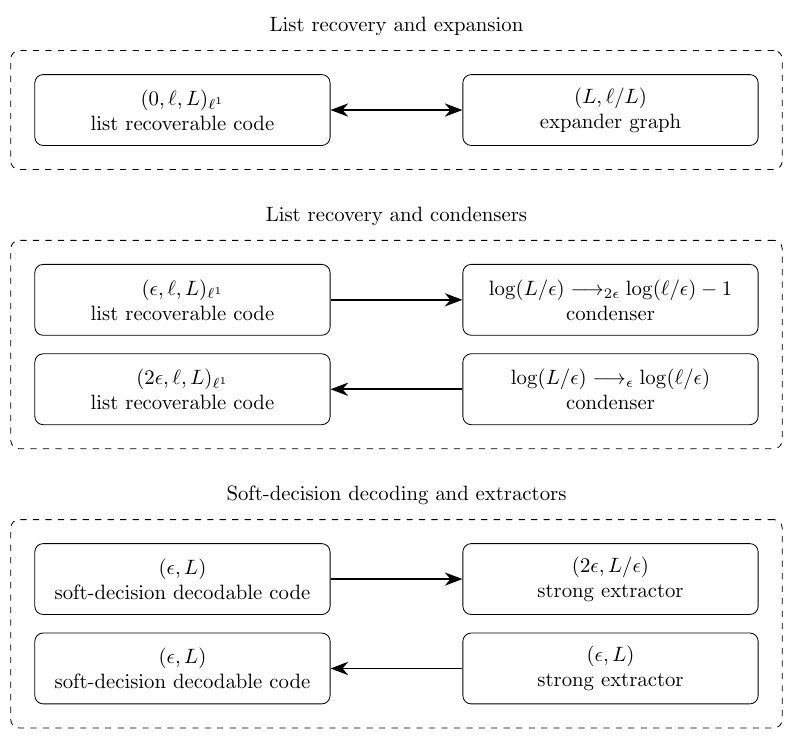}
		\caption{Correspondences between list recovery and expanders, condensers and extractors.  \textbf{Left:} The correspondence in~\cref{thm:expander-iff-LR}. \textbf{Middle:} Another correspondence in~\cite{guruswami-umans-vadhan-2009-unbalanced-expanders}. \textbf{Bottom:} Yet another correspondence in~\cite{TUZ01}.}\label{fig:LR-bijections}
	\end{figure*}
    
	\section{The unknown: open questions, incomplete answers, and more}

    While there has been great progress in recent years, a lot remains to be known.  We conclude this survey with a brief peek into the unknown.  We have mentioned several questions throughout this survey; we collect them here, and add a few more.
    \begin{enumerate}[(1)]
        \item  \emph{Better list recoverable AEL codes:}\quad We saw in~\cref{sec:AEL} that the AEL code construction can be used to obtain capacity achieving list decodable codes.  However, a simple-minded extension would yield fairly poor capacity achieving list recoverable codes, specifically with very large alphabet size.  Could this be remedied with some additional ideas?

        \item  (\cref{que:LB-small-alphabet}) \emph{Stronger lower bounds for linear codes over small fields:}\quad Can we establish an \(\ell^{\,\Omega(1/\epsilon)}\) lower bound on output list size for zero-error list recovery up to capacity of linear codes over small fields?

        \item (\cref{ques:list-size-small-alpha}) \emph{Stronger upper bounds for linear codes over small fields:} \quad Can we give better dependence of $L$ on $\ell$ for list-recovery from corruptions of linear codes close to capacity?

        \item  (\cref{que:L-poly-in-ell}) \emph{\(L=\ell^{\,O(1)}\):}\quad For large fields, if we consider the slightly weaker requirement of getting \emph{multiplicatively close} to list recovery capacity (with the underlying constant being arbitrarily close to 1), do there exist codes with the output list size being a fixed degree polynomial in the input list size?

        \item  (\cref{que:capacity-theorem}) \emph{Explicit codes close to the capacity theorem:}\quad  We are still not close to achieving the guarantees of the list recovery capacity theorem.  Find an explicit infinite family of codes (with increasing length) that achieves the guarantees of the capacity theorem.

        \item (\cref{ques:secret-sharing}) Can one prove that secret-sharing of lower rate ($<1/2$ for general linear codes, $<0.69$ for RS codes) are leakage-resilient. 

        \item  \emph{Beating David is tougher than beating Goliath:}\quad  Note that the guarantees of the capacity theorem kick in as soon as the length is a large enough constant -- in the proof of~\cref{thm:LR-capacity-theorem} that we present, we only need \(n\ge 9/\epsilon^2\), for a constant \(\epsilon>0\).  A terrific regime for further exploration is what we would like to call \emph{short asymptotics} -- assume \(\epsilon\) is a parameter approaching 0, and assume \(n=(1/\epsilon)^{O(1)}\).  Find explicit \emph{short length} codes that achieve the guarantees of the capacity theorem.  This regime seems more difficult to tackle, but could have more applications beyond coding theory.
    \end{enumerate}
    
    \emph{To infinity and beyond!}
    
	\section*{About the Authors}
	
	\hspace*{5mm}\tsf{Nicolas Resch} is an assistant professor with the Theoretical Computer Science Group from the Informatics Institute of the University of Amsterdam (UvA).  While he has broad interests in much of theoretical computer science, most of his work focuses on coding theory, cryptography, and their intersection.  Prior to joining the UvA, he was a postdoc in the Cryptology Group at the Centrum Wiskunde \& Informatica (CWI), hosted by Ronald Cramer. He obtained his Ph.D. from Carnegie Mellon University (CMU), and was advised by Venkatesan Guruswami and Bernhard Haeupler.
	
	\tsf{S. Venkitesh} is a postdoc at the Blavatnik School of Computer Science, Tel Aviv University, hosted by Amnon Ta-Shma.  He was earlier a postdoc at University of Haifa and IIT Bombay.  He obtained his Ph.D. from IIT Bombay, and was advised by Srikanth Srinivasan.  His research interests include error-correcting codes, pseudorandomness, algebraic methods in combinatorics, and Boolean function and circuit complexity.
	
	\printbibliography
	
	%\appendix
	%\appendixpage
	%\addappheadtotoc

\end{document}